\documentclass[reprint]{JASA}

\usepackage{amsmath, amssymb, amsthm}

\DeclareMathOperator*{\soft}{soft}
\DeclareMathOperator{\ba}{\mathbf{a}}
\DeclareMathOperator{\bd}{\mathbf{d}}
\DeclareMathOperator{\bc}{\mathbf{c}}
\DeclareMathOperator{\bA}{\mathbf{A}}
\DeclareMathOperator{\bF}{\mathbf{F}}
\DeclareMathOperator{\bI}{\mathbf{I}}
\DeclareMathOperator{\bv}{\mathbf{v}}
\DeclareMathOperator{\bw}{\mathbf{w}}
\DeclareMathOperator{\be}{\mathbf{e}}
\DeclareMathOperator{\bx}{\mathbf{x}}
\DeclareMathOperator{\by}{\mathbf{y}}
\DeclareMathOperator{\bu}{\mathbf{u}}
\DeclareMathOperator{\bs}{\mathbf{s}}
\DeclareMathOperator{\bS}{\mathbf{S}}
\DeclareMathOperator{\bD}{\mathbf{D}}
\DeclareMathOperator{\bH}{\mathbf{H}}

\newtheorem{theorem}{Theorem}

\begin{document}

\title{Approximate Extraction of Late-Time Returns via Morphological Component Analysis}

\author{Geoff Goehle} 
\email{goehle@psu.edu}
\author{Benjamin Cowen} 
\author{Thomas E. Blanford} 
\author{J. Daniel Park} 
\author{Daniel C. Brown}  
\affiliation{Applied Research Laboratory,  Pennsylvania State University, State College, PA, 16802, USA}

\begin{abstract}
	A fundamental challenge in acoustic data processing is to separate a measured time series into relevant phenomenological components. A given measurement is typically assumed to be an additive mixture of myriad signals plus noise whose separation forms an ill-posed inverse problem. In the setting of sensing elastic objects using active sonar, we wish to separate the early-time returns (e.g., returns from the object's exterior geometry) from late-time returns caused by elastic or compressional wave coupling.
	
	Under the framework of Morphological Component Analysis (MCA), we compare two separation models using the short-duration and long-duration responses as a proxy for early-time and late-time returns. Results are computed for Stanton's elastic cylinder model as well as on experimental data taken from an in-Air circular Synthetic Aperture Sonar (AirSAS) system, whose separated time series are formed into imagery. We find that MCA can be used to separate early and late-time responses in both cases without the use of time-gating.  The separation process is demonstrated to be robust to noise and compatible with AirSAS image reconstruction.  The best separation results are obtained with a flexible, but computationally intensive, frame based signal model, while a faster Fourier Transform based method is shown to have competitive performance.
\end{abstract}

\maketitle

\section{Introduction}
\label{sec:intro}

Underwater remote sensing using active sonar is typically performed by ensonifying the seafloor and processing the echoes to characterize the response from the objects and the environment. Synthetic Aperture Sonar (SAS) processing generates imagery of the scattering intensity of the ensonified scene, but typically assumes simplified acoustic scattering models akin to geometric optics. As such, the image formation algorithm only accounts for the early-time response of an object and innately couples the arrival time of acoustic energy with its spatial location. However, the overall response from an acoustically interrogated scene, especially the objects, supports additional responses including elastic scattering, structural resonances, and multi-path returns. These late-time responses do not conform to the image formation model, and are improperly associated to pixels during the image reconstruction process. They appear in the image as smearing or blurring \cite{Plotnick:2016a}, as seen in the top right plot of Figure~\ref{fig:airsas-chbf}. Additional artifacts arise due to the fact that the late-time signal structure has been spectrally modified by the acoustic coupling, structural vibration, and re-radiation back to the receiver. Differences in the signal structure elicit alternative processing approaches to help characterize the object. For example, localization of some late-time returns such as multiple scattering could be performed with some preprocessing steps before the image formation \cite{park-late}.

As a first step, we are motivated to separate the overall response into multiple components, each of which has some distinguishing property. In terms of signal decomposition, we hypothesize the overall response is a superposition of multiple components as well as noise at an unknown level. Decomposing this additive mixture into its components is an ill-posed inverse problem. Separation of the early-time and late-time returns from a non-homogeneous field of scatterers is a particularly challenging problem due to the diversity of acoustic effects \cite{pareige1989internal}.  While high-\(Q\) elastic responses such as whispering gallery modes produce long-duration ringing, low-\(Q\) modes such as surface wavepackets produce short-duration ringing, which can either arrive coincidentally with the geometrically scattered return or later in time \cite{kargl-shell}.  Multi-path returns can also arrive in the late-time, but are not considered in this work.

Various methods of separating early-time and late-time returns exist, from simple time gating to subtracting the response of a rigid object from an elastic one with an identical geometry.  A fundamental challenge with time-gating is that the early-time and late-time responses from a field of scatterers will overlap in time, preventing a clean separation of components.  Subtracting the responses of objects with the same geometry but different material properties is an important tool, but limited to analytic or laboratory settings.  In this paper, we approach this problem using a relatively recent technique in convex optimization known as Morphological Component Analysis (MCA)~\cite{ogMCA1,ogMCA2,selesSALSA}. MCA is an optimization framework where each component of an additive mixture is identified by its ability to be sparsely represented by a unique linear operator, such as a frame or a dictionary. In this paper, we present two sparse representation frameworks for discriminating acoustic phenomena and compare their performance.

While we are motivated by the separation of the early-time and late-time responses, the paper will focus on the related problem of separating the short-duration and long-duration components of a time series.  In this context, a short duration component is any signal component which has a short time duration, regardless of the physical source.  This will include both the initial geometrically scattered return from the object as well as any late-time wavepackets resulting from surface coupling.  Long duration components will generally include long tailed exponential decays caused by high-\(Q\) resonance modes.  The reason we focus on short-duration/long-duration separation versus early-time/late-time is that we are motivated by the application to sonar imaging, where time series feature multiple superimposed returns with start times that are not known {\em a priori}.  As such none of the separation techniques presented here will rely on time gating.

Section~\ref{sec:theory} describes the MCA framework and the two sparsification transforms featured in this paper.  In Section~\ref{sec:stanton} we use MCA to separate the short-duration and long-duration components of an analytic time series produced by Stanton's elastic cylinder model \cite{stanton2}. In Section~\ref{sec:airsas} we apply the same MCA techniques to experimental data collected using an in-Air circular Synthetic Aperture Sonar (AirSAS) \cite{airsas} and demonstrate short-duration/long-duration separation on AirSAS imagery. Section~\ref{sec:discussion} concludes with a discussion of MCA as applied to acoustic time series.

\subsection*{Notation}
Brackets are used to denote the scalar elements of vectors, e.g. for $\by\in\mathbb{C}^N$ we have
\linenomath
\begin{equation*}
	\by = \left[y[0], y[1], ..., y[N-1]\right].
\end{equation*}
Subscripts are used to differentiate the vectors, matrices, and parameters associated with the distinct components of the signal that we wish to separate. For example, $\by = \by_1+ \by_2$ denotes a vector signal $\by$ composed of two vector components $\by_i$, each of which may have an associated scalar parameter $\lambda_i$ and matrix parameter $\bA_i$ for $i=1,2$.

\section{Morphological Component Analysis}
\label{sec:theory}
For the following analysis we assume that our measured acoustic data $\by$ is an additive mixture of $D$ morphologically distinct components:
\linenomath
\begin{align}
	\by = \sum_{i=1}^D \by_i.
\end{align}
The recovery of these components is an ill-posed inverse problem because there are infinitely many trivial solutions.
The MCA framework~\cite{ogMCA1} addresses this issue by requiring each component $\by_i$ to admit sparse representation $\bx_i$ in a corresponding transformed space given by a linear operator $\bA_i:\mathbb{C}^{M_i}\to\mathbb{C}^N$. In this context, each $\bA_i$ is called a synthesis operator because it synthesizes coefficients $\bx_i$ into the signal domain. We can write the MCA signal model as
\begin{align}
	\by = \sum_{i=1}^D \bA_i \bx_i,
\end{align}
under which the problem of signal separation becomes a convex optimization problem: we want to find sparse encodings $\bx_i$ such that the original data is preserved. For $D=2$, the optimization over sparse coefficients $\bx_i$ is written as
\begin{equation}
	\label{eq:mca-bp}
	\begin{aligned}
		\text{arg}\min_{\bx_1, \bx_2}\ &\lambda_1 \|\bx_1\|_1 + \lambda_2\|\bx_2\|_1 \\
		\text{s.t. } \by &= \bA_1\bx_1 + \bA_2\bx_2.
	\end{aligned}
\end{equation}
This minimization problem identifies a sparse set of coefficients \(\bx_i\) so that the original signal $\by$ can be exactly reconstructed. The $\ell_1$-norm, $\|\bx\|_1 = \sum_n |x(n)|$, enforces sparsity in a minimization context by penalizing all non-zero components and by introducing a thresholding operation in the optimization algorithms discussed herein that sets small values to zero when possible. The $\lambda_i \in \mathbb{R}_+$ are tunable parameters that affect the severity of penalizing non-zero coefficients and can be used to prioritize one representation over the other. This problem is also referred to as Dual Basis Pursuit (BP)~\cite{selesSALSA,ogBP}.  The equality constraint can be relaxed to perform denoising, a problem known as Dual Basis Pursuit Denoising (BPD), and in this case is represented by
\begin{align}
	\label{eq:mca-bpd}
	\begin{aligned}
		\text{arg}\min_{\bx_1, \bx_2}\ & \begin{array}{l} \lambda_1 \|\bx_1\|_1 + \lambda_2\|\bx_2\|_1 + \\ \frac{1}{2}\| \by - \bA_1\bx_1 - \bA_2\bx_2 \|^2. \end{array}
	\end{aligned}
\end{align}
For BPD the \(\lambda_i\) parameters control both the weights applied to the encodings as well as the degree to which sparsity is prioritized over fidelity, with larger values of \(\lambda_i\) producing sparser, less accurate, reconstructions.

In MCA, our ability to identify a component via its sparse representation hinges on the aptness and mutual exclusivity of each linear operator. In other words, each transform $\bA_i$ should admit sparse representation of its corresponding component signal $\by_i$, but should be inefficient in representing the other components. We wish to decompose \(\by\) into a sum of short-duration components \(\by_1\) and long-duration components \(\by_2\) (as proxies for the early-time and late-time returns as discussed in Section~\ref{sec:intro}), and thus the problem at hand is to design $\bA_1$ and $\bA_2$ to describe those respective phenomena. We focus on over-complete tight-frame operators $\bA_i$~\cite{framesforundergraduates} which, by definition, satisfy \(\bA_i\bA_i^* = p_i\bI\) for \(p_i > 0\).  The subsequent sections discuss specific, promising selections of $\bA_i$ for our application.

Problems~\ref{eq:mca-bp} and~\ref{eq:mca-bpd} are convex and hence have unique, global solutions~\cite{boydConvex,selesSALSA}, but the solutions do not have a closed-form expression due to the non-smooth $\ell_1$-norm. We can use the Alternating Direction Method of Multipliers (ADMM) to formulate these problems as a sequence of easier subproblems, whose iterative solution is guaranteed to converge to the global minimum~\cite{admmConvergence}. The resulting algorithm is called the Split Augmented Lagrangian Shrinkage Algorithm (SALSA)~\cite{salsa, selesSALSA}. SALSA as applied to our MCA BP and BPD problems is written in Algorithm~\ref{alg:mca-bp}, where \(\soft\) represents the soft-thresholding function
\[
\soft(x, T) = \begin{cases} \frac{|x|-T}{|x|} x & |x| > T \\ 0 & |x| \leq T. \end{cases}
\]
Note that the difference between BP and BPD in Algorithm~\ref{alg:mca-bp} is a single constant.  While \(\bx_i\) converges to the solution, in the BP case it is not particularly sparse at any given iteration.  As an alternative, \(\bu_i\) also converges to the solution while being sparser at each iteration.

\begin{algorithm}
\caption{MCA BP/BPD}
\label{alg:mca-bp}
\begin{algorithmic}
\REQUIRE \(\by\), \(\bA_i\), \(\lambda_i\), \(\mu\)
\STATE initialize \(\bx_i = \bA_i^*\by\), \(\bd_i = 0\), \(i = 1,2\)
\IF{performing BP}
\STATE \(\alpha = \frac{1}{p_1 + p_2}\)
\ELSE
\STATE \(\alpha = \frac{1}{\mu + p_1 + p_2}\)
\ENDIF
\REPEAT
\STATE \(\bu_i \gets \soft(\bx_i + \bd_i, \lambda_i/\mu)\), \(i = 1,2\)
\STATE \(\bv_i \gets \bu_i - \bd_i\), \(i=1,2\)
\STATE \(\bc \gets \by - \bA_1\bv_1 - \bA_2\bv_2\)
\STATE \(\bd_i \gets \alpha \bA_i^*\bc\), \(i = 1,2\)
\STATE \(\bx_i \gets \bd_i + \bv_i\), \(i = 1,2\)
\UNTIL{stopping criteria met}
\STATE \(\by_i \gets \bA_i\bx_i\), \(i = 1,2\)
\end{algorithmic}
\end{algorithm}

The \(\lambda_i\) scalars act as a weighting factor influencing how energy is prioritized between the \(\bx_i\) and, in the case of BPD, how much sparsity is prioritized over reconstruction fidelity.  While in practice it is often useful to tune the \(\lambda_i\) to achieve a desired separation, in order to support comparative analysis for this paper we will use a common value \(\lambda = \lambda_1  = \lambda_2\).  The choice of common \(\lambda\) does not effect the solution for BP, but is important for BPD.  When performing BPD there is a maximum effective \(\lambda\)-value given by
\[
\lambda_{\max} = \max( \| \bA_1^* \by\|_\infty, \|\bA_2^*\by\|_\infty)
\]
such that for all \(\lambda \geq \lambda_{\max}\) the solution is zero \cite[Section V.B]{compressive}.  We will generally choose \(\lambda\) as a percentage of \(\lambda_{\max}\).  Additionally, while it is possible to vectorize \(\lambda\) to achieve even finer grained control over the separation weights we will not do so here.

\subsection{FFT MCA}
\label{sec:theory-fft}

A particularly simple, yet effective, form of MCA is to let the first representation be the identity, \(\bA_1 = \bI\), and the second be the unitary Discrete Fourier Transform, \(\bA_2 = \bF\).  We refer to this as FFT MCA and in this case the solution to \eqref{eq:mca-bp} or \eqref{eq:mca-bpd} splits a signal into two components, with the former sparse in the time domain and the latter sparse in the frequency domain.  A consequence of Fourier duality is the \(\by_1\) component tends to be made up of broadband, short-duration elements while \(\by_2\) tends to be made up of long-duration elements with a narrower spectrum.

A classic example is to consider the superposition of a spike on a sinusoid.  Suppose we have \(N=1000\) samples with \(f_s = 10\)kHz and define \(\by\) to be
\linenomath
\[
y[n] = \delta_{50}[n] + \sin\left(2\pi \frac{1000}{f_s} n\right)
\]
where \(\boldsymbol{\delta}_{50}\) is a one-hot vector at index 50.  If BPD is applied to \(\by\) using FFT MCA then Algorithm~\ref{alg:mca-bp} will converge to
\begin{align}
\label{eq:fftsig}
y_1[n] &= \delta_{50}[n], & y_2[n] &= \sin(2\pi 1000 n / f_s).
\end{align}
In this case MCA separates \(\by\) into its components exactly.  This would not be true if, for example, noise were added to \(\by\) or if \(\by\) was the superposition of a sinusoid and a rectangular window.  That is because noise and/or rectangular windows are not sparse with respect to either \(\bI\) or \(\bF\).  For a noisy signal, the correct approach would be to use BPD to reconstruct the signal without the noise component.  For the rectangular window, exact separation is easier using a different set of representations.  As FFT MCA is signal agnostic and has no parameters, there isn't any way to alter the representations to fit a particular signal model.

\subsection{ESP MCA}
\label{sec:theory-esp}

A more flexible set of representations are given by Enveloped Sinusoid Parseval (ESP) frames, a class of representations formed from enveloped and shifted sinusoids.  Formal derivation of ESP frame theory is presented in Appendix~\ref{apx:esp}. Briefly, given a set of non-zero (but potentially complex) envelopes \(\{\be_l\}_{l=0}^{L-1}\subset\mathbb{C}^{N}\) the vectors \(\{\ba_{l,k,m}\}\) defined by
\[
a_{l,k,m}[n] = e_l[n-m \bmod N]\exp(2\pi j k(n-m)/N)
\]
for \(l=0,\ldots, L-1\) and \(k,m,n=0,\ldots,N-1\) form a tight frame.  Here \(l\) is the envelope index, \(k\) is the frequency index, and \(m\) is the time shift index.  The synthesis operator \(\bA\) in this case is given by \(A[n,l,k,m] = a_{l,k,m}[n]\).  Since ESP frames can be made using nearly arbitrary envelopes, a wide range of functions can be sparsely represented including exponentially decaying sinusoids, sinusoids with traditional windows, or modulated complex signals. However, as there are \(N^2L\) frame vectors, ESP frames are massively overdetermined.  One of the advantages of SALSA is that it is not necessary to work with the synthesis and analysis matrices directly.  Instead Appendix~\ref{apx:esp} describes FFT diagonalization which can be used to speed up Algorithm~\ref{alg:mca-bp}.

The goal when applying ESP frames to MCA is to find two sets of envelopes \(\be_l^i\), where the superscript indicates the component index and the subscript the envelope index, such that the signal components \(\by_i\) are sparsely represented by one set of frame vectors but not the other.  In the ideal case, \(\by_i\) is actually equal to a frame vector for \(\bA_i\).  The specific choice of envelope is often informed by the physics associated to the signal in question.  In this case, we wish to separate the long-duration high-\(Q\) signal components from the short-duration acoustic response of an elastic object.  As such, we will use decaying exponentials as envelopes for \(\bA_2\) since exponentially decaying sinusoids are an excellent signal model for long-duration ring down \cite{hambric2006}.  For \(\bA_1\) we will use extremely short rectangular windows since they flexibly capture short-duration signals.

As an aside, if \(\bA_1\) is generated using a single one-hot vector as an envelope, while \(\bA_2\) is generated using a single constant function as an envelope, then the resulting representations are extremely similar to the representations used in FFT MCA.  This mode of ESP MCA effectively generalizes FFT MCA, albeit not in strict mathematical terms.  For example, if this degenerate ESP frame and FFT MCA are both applied to the signal described in \eqref{eq:fftsig} using BP with 1000 iterations, then the relative difference in the resulting \(\by_1\) and \(\by_2\) components is 0.077\% and 0.879\%, respectively.

\subsubsection*{ESP MCA Example}

In order to illustrate the link between ESP frame envelopes and the underlying signal morphology we will perform ESP MCA separation using the following driven simple harmonic oscillator
\begin{equation}
\label{eq:ode}
\ddot{y} + \frac{2}{\tau} \dot{y} + 4\pi^2 f_0^2 y = \alpha \sin(2\pi f t)
\end{equation}
with \(\tau = 2\)ms, \(f = 15\)kHz, \(f_0 = 20\)kHz, and \(\alpha = 10^{10}\). Let \(y\) be the zero state solution (\(y(0) = \dot{y}(0) = 0\)). We generate \(\by\) using \(N=1000\) samples of \(y\) with a sampling frequency of \(f_s = 100\)kHz.  Based on what we know of the underlying dynamics, we can choose envelopes that will allow us to exactly capture the decomposition of \(\by\) into homogeneous and particular components.   Since the homogenous solution to \eqref{eq:ode} must consist of exponentially decaying sinusoids we can let \(\bA_1\) be the ESP frame associated to a single exponentially decaying envelope \(e^1_1[n] = \exp(-n/(\tau f_s))\) and expect \(\bA_1\) to capture the homogeneous part of \(\by\).  Similarly the particular solution to \eqref{eq:ode} must be a sinusoid so we define \(\bA_2\) to be the ESP frame associated to the single constant envelope \(\be^2_1 = 1\) and expect \(\bA_2\) to capture the particular component.  If we then perform MCA BP with 1000 iterations we get the separated signals shown in Figure~\ref{fig:ode}. As we used knowledge of the ODE dynamics to define ESP frames such that \(\by_i\) is sparsely represented by \(\bA_i\), the BP algorithm converges to the exact separation of \(\by\) into its homogeneous and particular parts with the relative error in this case equal to 0.12\% for the homogeneous solution \(\by_1\) and 0.05\% for the particular solution \(\by_2\). Notably the FFT MCA approach would not be able to achieve the exact separation presented here because the homogeneous solution is not sparsely representable using the FFT.

\begin{figure}
\centering
\includegraphics[width=3.25in]{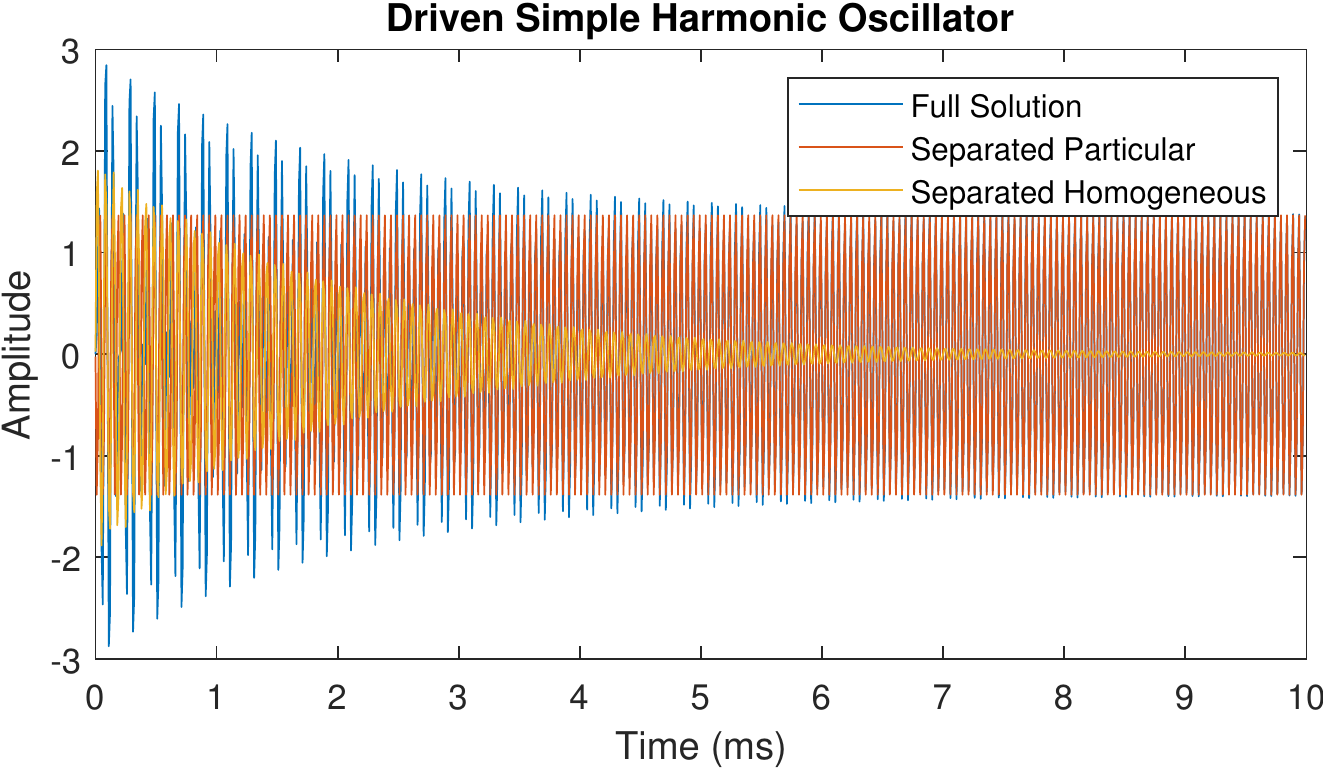}
\caption{(color online) ESP MCA separated homogeneous and particular components for the zero state solution to \eqref{eq:ode} using BP with 1000 iterations.  The relative error for the particular solution is 0.05\% and the homogeneous solution is 0.12\%.}
\label{fig:ode}
\end{figure}

\subsection{MCA of Acoustic Signals}
\label{sec:theory-reps}

We are generally interested in separating out the short-duration returns of an elastic object from the long-duration ones.  From a physical perspective, the short-duration returns include the initial return of the ping reflecting off the rigid geometry of the object, as well as low-\(Q\) elastic effects and additional short-duration late-time phenomena.  The long-duration returns primarily include the high-\(Q\) resonance modes of the object.  From a signal analysis perspective, however, the specific form of what constitutes a short-duration or long-duration component is ultimately defined by the MCA representations.
For FFT MCA the short-duration returns are represented using one-hot vectors, since \(\bA_1 = \bI\), while the long-duration returns are represented using sinusoids, since \(\bA_1 = \bF\).
ESP MCA will use a frame built from short rectangular windows to capture the short-duration components (representing them as very short windowed sinusoids) and a frame built from exponentially decaying envelopes to capture the long-duration components (representing them as exponentially decaying sinusoids).


\section{Analytic Signal Separation}
\label{sec:stanton}

In this section we will demonstrate the MCA approaches presented in Section~\ref{sec:theory} on an analytic acoustic signal produced by the Stanton elastic cylinder model \cite[Section B]{stanton2}. Section~\ref{sec:stanton-impulse} demonstrates separation applied to a clean impulse response from a Stanton elastic cylinder.  In Section~\ref{sec:stanton-lfm} we demonstrate the same separation on a noisy LFM response while in Section~\ref{sec:stanton-noise} we analyze the performance of the MCA techniques over a range of noise levels.

\subsection{Impulse Response Separation}
\label{sec:stanton-impulse}

The Stanton model parameters used in this paper were chosen to represent a solid aluminum cylinder in water with a diameter of 15.25cm and a length of 30.5cm.  The receiver is located 2m from the cylinder with a centered broadside orientation, and the signal is sampled at \(f_s = 300\)kHz.  Stanton's model provides the frequency representation of the cylinder's impulse response.  In order to minimize non-causal effects we apply a Butterworth filter of order 3 and threshold 0.25 to the synthesized time series.  This helps to remove spectral discontinuities and produces a more natural impulse.  The resulting time series, and corresponding spectrum, are shown in Figure~\ref{fig:stanton-impulse}.  The deep nulls at 15kHz, 23kHz and 30kHz are likely caused by low-\(Q\) surface wave elastic effects.  These effects are short duration, and may interfere constructively or destructively to the geometric scattering response.  The intention is for these effects to be included in the short-duration component.  The sharper, shallower nulls at 18kHz, 28kHz, 39kHz and 42kHz correspond to high-\(Q\) geometric resonance modes.  These are long duration signals and are one of our primary targets for the long-duration component.

\begin{figure}
\centering
\includegraphics[width=3.25in]{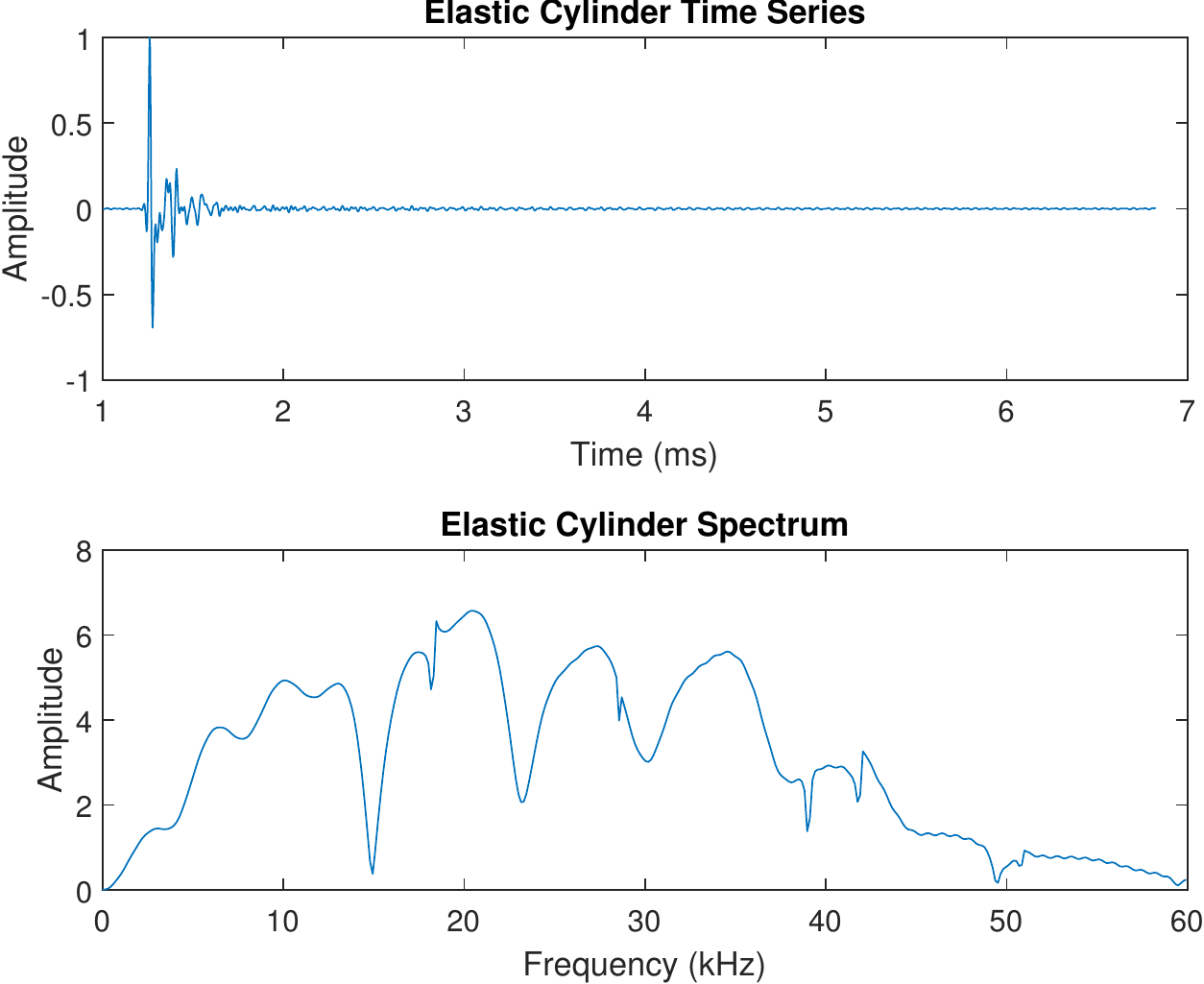}
\caption{(color online) Impulse response for the Stanton elastic cylinder (top) and corresponding spectral power (bottom).}
\label{fig:stanton-impulse}
\end{figure}

Importantly, since this model features a single return with no noise, the short/long-duration separation problem is largely equivalent to early/late-time separation.  Consequently will use the relative error between the original time series and the separated short-duration component in the early time, as well as the relative error between the time series and the long-duration component in the late time, as performance metrics for the MCA techniques.  Specifically we fix an early-time interval \(I_1\) lasting from 1ms to 2ms which includes the initial signal return, and a late time interval \(I_2\) lasting from 2ms to 6ms which only includes the late signal return.  Our metrics are then defined in terms of the standard \(\ell_2\)-norm by
\begin{align}
\label{eq:metrics}
m_1 &= \frac{\left\| \by|_{I_1} - \by_1|_{I_1} \right\|}{\left\| \by|_{I_1}\right\|}, &
m_2 &= \frac{\left\| \by|_{I_2} - \by_2|_{I_2} \right\|}{\left\| \by|_{I_2}\right\|}
\end{align}
where \(\by|_{I}\) indicates the restriction of \(\by\) to the interval \(I\). In the case where the separation is exact, we would expect \(m_2\) to be zero, since the late-time signal contains only long-duration components, but \(m_1\) to still be nonzero since the early-time signal does contain some long-duration energy.  This effect will be minor since the signals presented in this section are dominated by short-duration energy.

\subsubsection*{FFT MCA}

We begin by applying FFT MCA to the Stanton signal using BP with 1000 iterations.  After performing the separation, we get the results shown in Figure~\ref{fig:stanton-fft}. There we have plotted the original time series, the short-duration component, and the long-duration component in the early time, in the late time, and in the frequency domain.  The results are quite good.  FFT MCA correctly separates the loud initial response into the short-duration component while the signal tail is entirely separated into the long-duration component.   Quantitatively, the short-duration component has an early-time error of 4.44\% while the long-duration component has a late-time error 3.70\%.  The behavior of the spectrum is particularly interesting.  The bulk of the spectral power for the impulse response is separated into the short-duration signal, including the wide nulls caused by the low-\(Q\) elastic responses.  The sharp short high-\(Q\) nulls however have been turned into distinct spikes in the spectrum of the long duration component.  This has implications for feature detection as narrow peaks in the frequency domain are easier to detect and more resilient to noise than narrow nulls.

\begin{figure}
\centering
\includegraphics[width=3.25in]{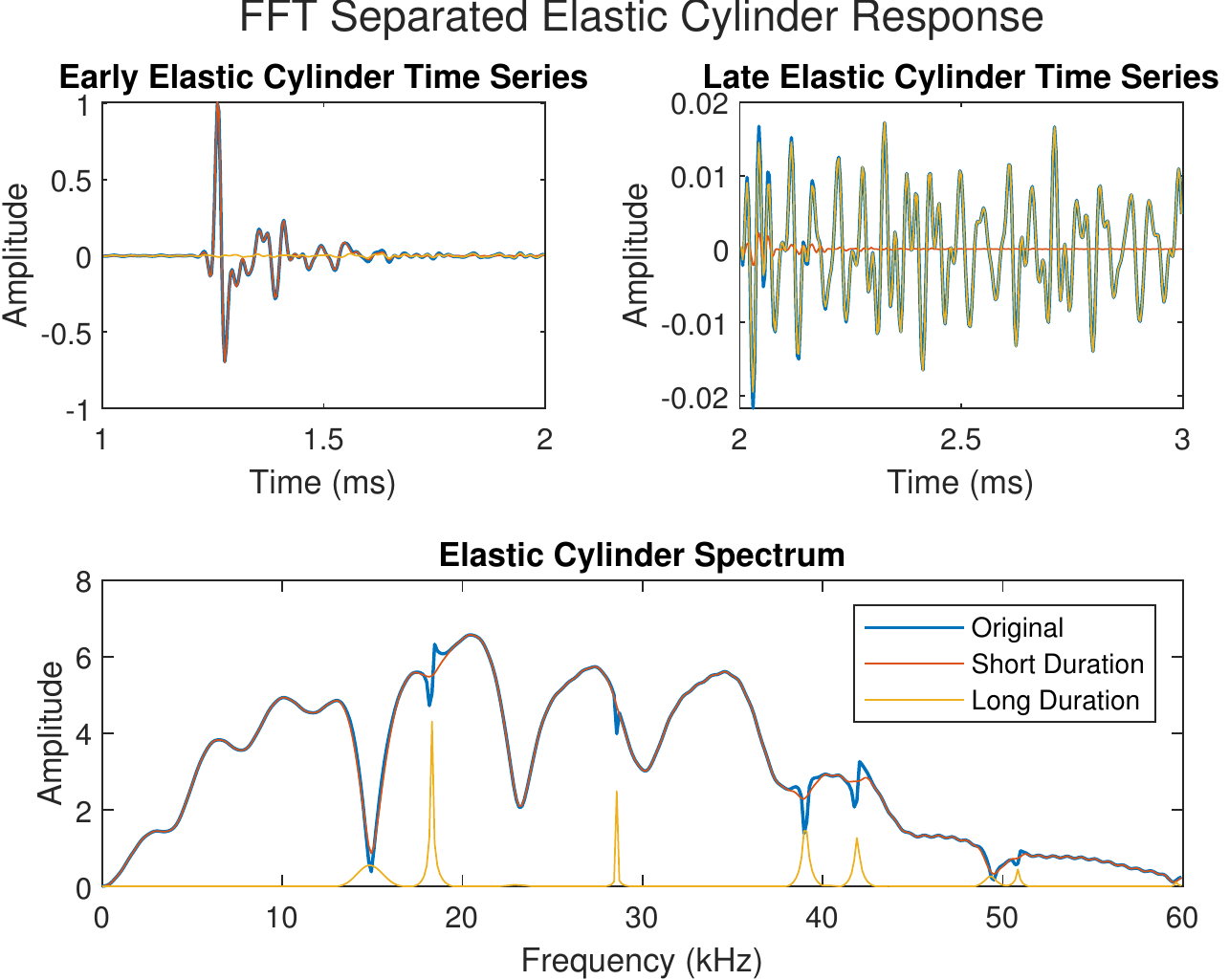}
\caption{(color online) Impulse response as well as the FFT separated short and long duration components for the elastic cylinder (top) and corresponding spectral power (bottom).  The short-duration component has an early-time relative error of 4.44\% and the long-duration component has a late-time relative error of 3.70\%.}
\label{fig:stanton-fft}
\end{figure}

\subsubsection*{ESP MCA}

As discussed in Section~\ref{sec:theory-esp}, for ESP separation we are using two categories of envelopes.  In this section \(\bA_1\) will be formed using rectangular windows
\[
e^1_l(t) = \begin{cases} 1 & 0 < t < T_l \\ 0 & \text{otherwise.} \end{cases}
\]
The representation \(\bA_2\) will be formed using exponentially decaying envelopes
\[
e^2_l(t) = \exp(-t/\tau_l).
\]
We will use logarithmically spaced window lengths and time constants given by
\begin{align}
\label{eq:esp-param}
T_l &= 0.27, 0.54, 0.1\text{ms}, \\
\tau_l &= 1.78, 3.16, 5.62, 10.00, 17.78, 31.62\text{ms} \nonumber
\end{align}
Several different factors were considered in the selection of these parameters:
\begin{itemize}
\item The longest window length, 0.1ms, is significantly shorter than the shortest time constant, 1.78ms.  This ensures the atoms are morphologically distinct and encourages better separation.
\item The shortest window length, 0.27ms, is long enough to support a significant number of oscillations in the frequency ranges of interest.
\item The largest time constant, 31.62ms, is big enough to support envelopes which decay very little over the length of the signal.
\item The shortest time constant, 1.78ms, still produces atoms which would be considered long-duration.
\end{itemize}

\begin{figure}
\centering
\includegraphics[width=3.25in]{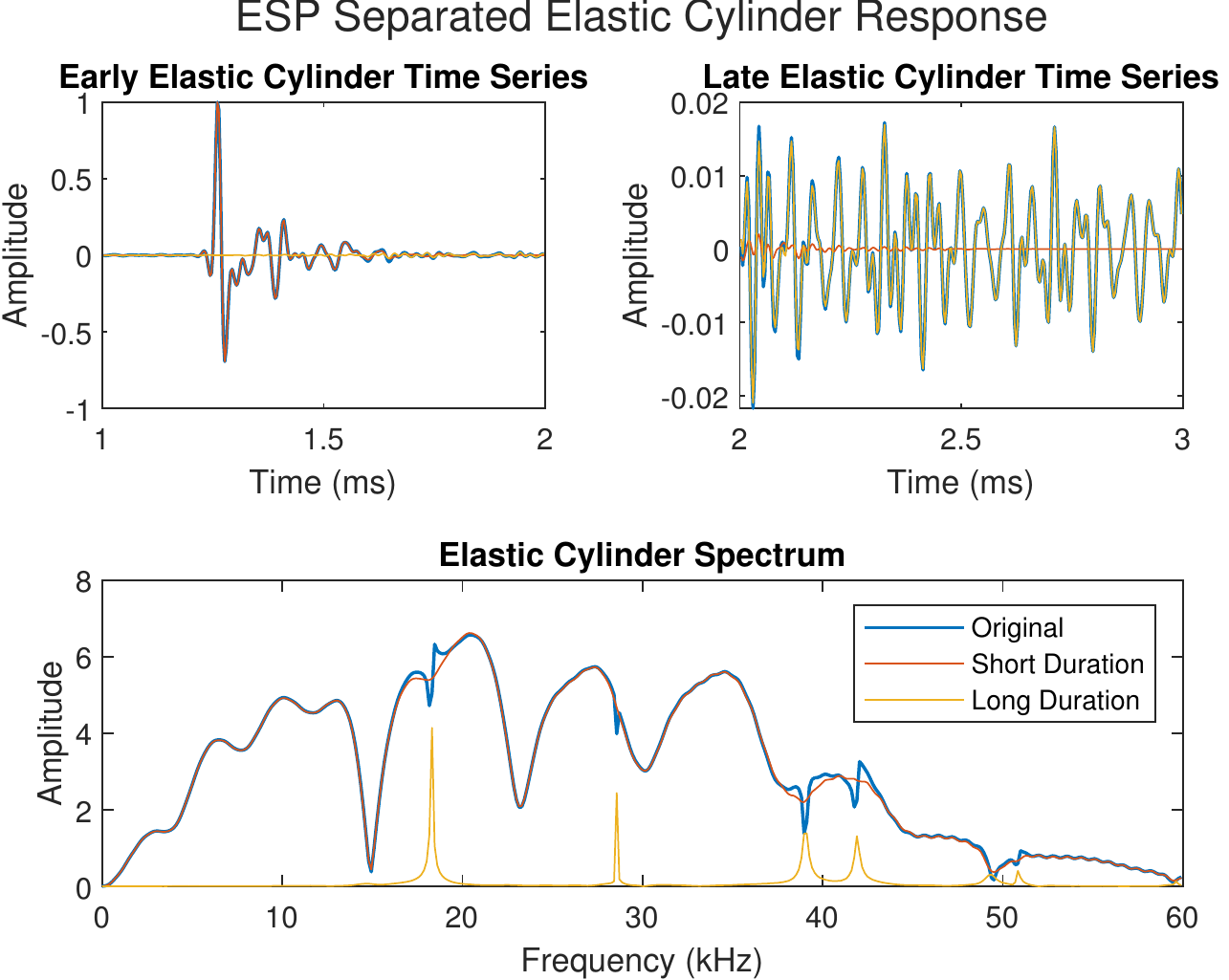}
\caption{(color online) Impulse response as well as the ESP separated short and long duration components for the elastic cylinder (top) and corresponding spectral power (bottom). The short-duration component has an early-time relative error of 4.04\% and the long-duration component has a late-time relative error of 3.67\%.}
\label{fig:stanton-esp}
\end{figure}

The results of ESP MCA utilizing 1000 iterations of BP are shown in Figure~\ref{fig:stanton-esp}.  For this time series, ESP MCA closely mirrors FFT MCA, and is nearly as effective at separating the signal as in the toy ODE case.  We see that the short duration component captures almost all of the initial response, and most of the spectral power, with an early-time error of 4.04\%.  The long-duration component on the other hand captures the entire late-time response, with a late-time error of 3.67\%, and as a result has some clear peaks at the high-\(Q\) null locations.  Comparing the performance of both methods in Table~\ref{table:stanton} we see that FFT and ESP MCA perform about the same in terms of relative error, which is confirmed by a visual inspection of the separated components.  

\begin{table}
\centering
\begin{tabular}{c | c c }
~ & \multicolumn{2}{c}{Relative Error} \\
Method & \(m_1\)  & \(m_2\) \\ \hline
FFT MCA & 4.44\% & 3.70\%  \\
ESP MCA & {\bf 4.04\%} & {\bf 3.67\%}
\end{tabular}
\caption{Short-duration early-time relative error \(m_1\) and long-duration late-time relative error \(m_2\) for FFT MCA and ESP MCA for the noise free impulse response of the Stanton elastic cylinder.}
\label{table:stanton}
\end{table}

\subsection{Noisy LFM Response Separation}
\label{sec:stanton-lfm}

In order to understand the impact of noise on our MCA techniques, and to represent more realistic signal processing, we will now use an LFM response from the same Stanton cylinder model to produce a noisy, matched-filtered, time series.  The clean LFM scattering response was generated by convolving an LFM which sweeps from 15kHz to 45kHz over 1ms with the Stanton model impulse response presented in Section~\ref{sec:stanton-impulse}, again using a 300kHz sampling frequency.  We also use the same Butterworth filter as the previous section, but because the LFM already effectively band pass filters the signal the effect is minimal.  We generate a noisy LFM return by adding white Gaussian noise at a 10dB SNR measured against the average signal power.  Lastly, matched filtering is applied to produce the final clean and noisy time series shown in Figure~\ref{fig:stanton-lfm}.  The added noise is more apparent in the spectrum where it mostly obscures the sharp high-\(Q\) resonance nulls (18kHz and 28kHz especially).

\begin{figure}
\centering
\includegraphics[width=3.25in]{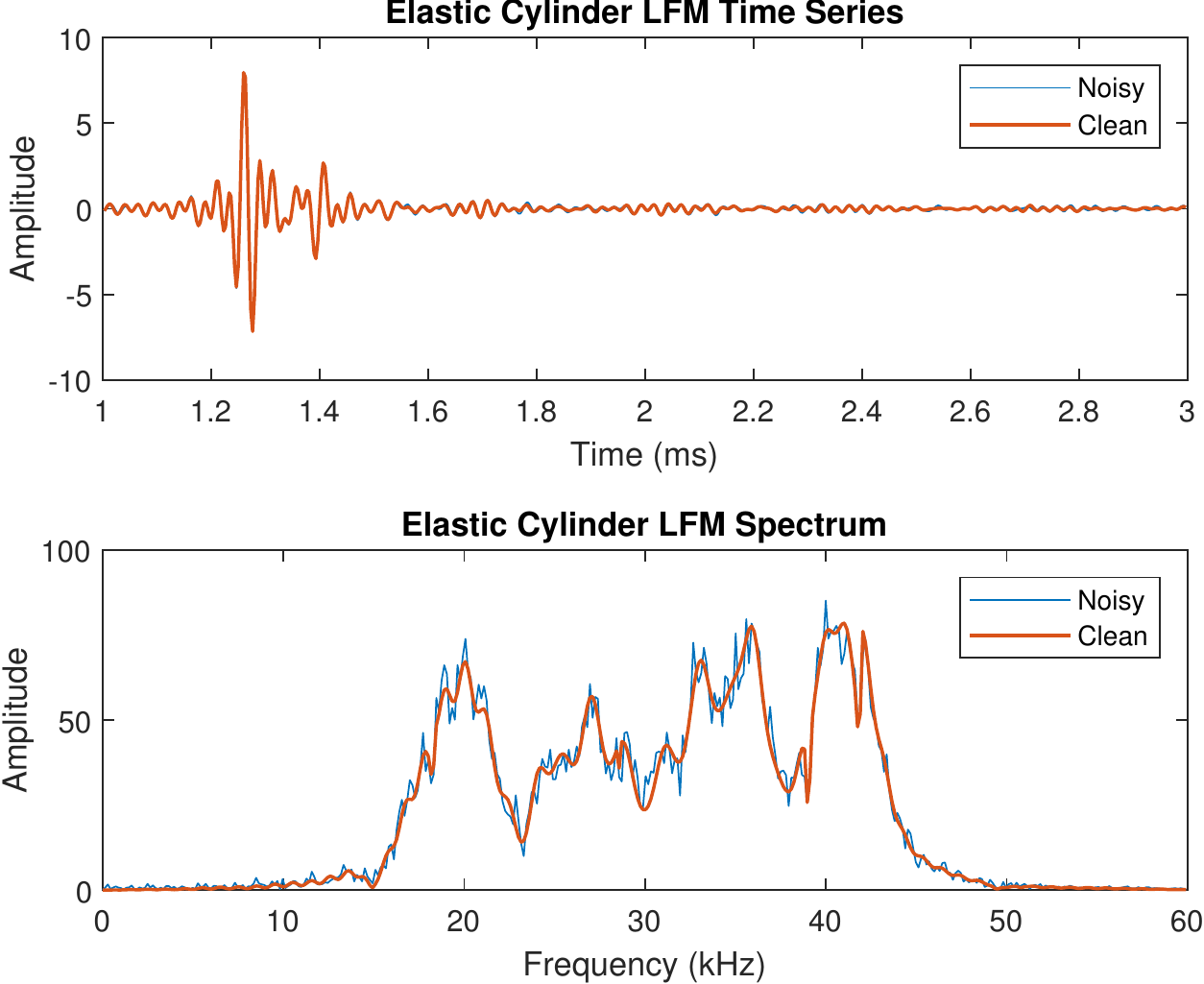}
\caption{(color online) Noisy and clean matched filtered LFM response for the elastic cylinder (top) and corresponding spectral power (bottom).}
\label{fig:stanton-lfm}
\end{figure}

In this section we will again use the metrics defined in \ref{eq:metrics}, but measured against the clean signal.  Notably the early-time portion of the LFM return is approximately 20dB louder than the late-time portion.  As the SNR of the noisy signal is measured relative to the average signal power, which is dominated by the high power in the early return, the noise is much louder when compared to the quieter late-time return.  Specifically, the post-matched filtered SNR of the noisy time series in the early-time is 25dB (where it is benefiting the most from filtering) while the post-matched filtered SNR of the noisy time series in the late-time is -1dB.

\subsubsection*{FFT MCA}

As our signal contains broadband noise we will utilize MCA BPD, which allows for some amount of reconstruction error in order to reduce the amount of noise in the separated signals.  Using FFT MCA and BPD with 1000 iterations and \(\lambda = 0.01\lambda_{\max}\) we have the results in shown Figure~\ref{fig:stanton-noise-fft}.  The top plots of Figure~\ref{fig:stanton-noise-fft} compare the noisy signal, the clean signal, and the separated short-duration and long-duration components in both the early-time and the late-time while the bottom plot shows the associated spectra.  The addition of noise, and the use of an LFM, has had a large impact on the quality of the MCA separation.  The early-time error for the short-duration component is 17.6\%, and the late-time error for the long-duration component is 75.4\%.  While the short-duration component looks good in the time domain it is clear from the frequency domain that a significant amount of spectral information has been lost.  The long-duration spectrum is also significantly muddled. While the 18kHz peak is still visible, most of the other expected resonance peaks are lost in the noise.

\begin{figure}
\centering
\includegraphics[width=3.25in]{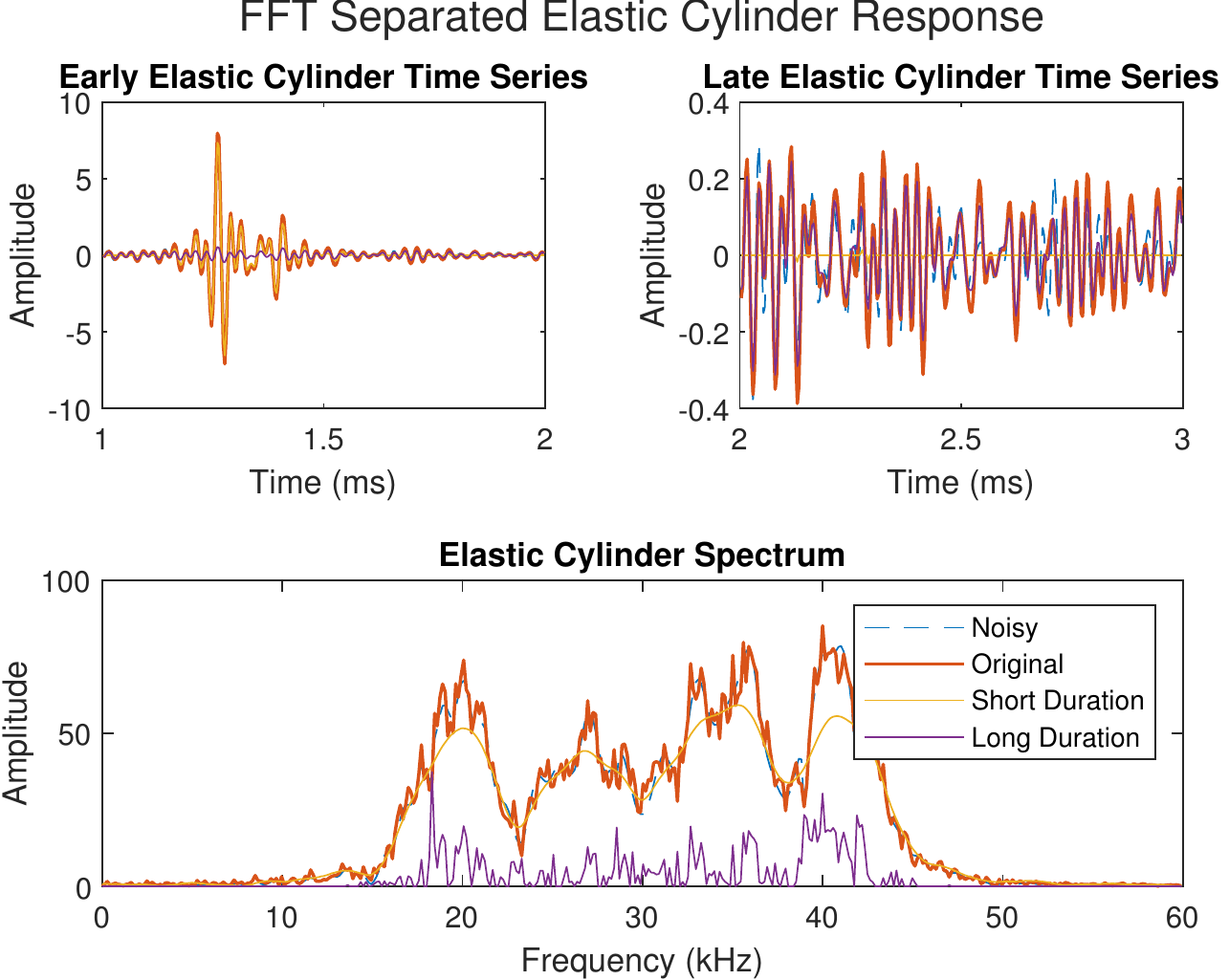}
\caption{(color online) Clean matched filtered LFM response, noisy matched filtered LFM response and the FFT separated short and long duration components for the elastic cylinder (top) and corresponding spectral power (bottom).  Noise was added at a 10dB SNR level.  The short-duration component has an early-time relative error of 17.6\% and the long-duration component has a late-time relative error of 75.4\%.}
\label{fig:stanton-noise-fft}
\end{figure}

\subsubsection*{ESP MCA}

Next we separate the noisy LFM signal using ESP MCA.  While we still use rectangular and exponentially decaying envelopes we will use window lengths and time constants of
\begin{align*}
T_l &= 0.07, 0.27, 0.54\text{ms}, \\
\tau_l &= 3.16, 5.62, 10.00, 17.78, 31.62\text{ms}.
\end{align*}
The main difference between the envelope parameters for this section and Section~\ref{sec:stanton-impulse} is that we have removed the 0.1ms rectangular window, added a 0.07s window, and dropped the 1.78s time constant.  This change is a heuristic choice driven by the characteristics of the LFM signal.  Increasing the gap between the longest rectangular window and the smallest exponential time constant decreases the ``overlap'' between the two ESP frames and tends to make the separation more robust to noise at the cost of performance in the noise-free case.  While this illustrates the flexibility of ESP frames, establishing a formal approach for determining optimal envelope parameters is an open question.

If we perform ESP MCA using BPD with 1000 iterations and \(\lambda = 0.01\lambda_{\text{max}}\) we get the results in Figure~\ref{fig:stanton-noise-esp}.  We can see visually that while the separation is still largely successful, particularly in the time domain, the noise and LFM have impacted the clarity of the spectral peaks in the resonant response.  The short-duration early-time error is 11.1\%.  The spectrum of the short-duration component is a better fit to the LFM spectrum than the FFT MCA short-duration component, but is still lacking definition.  The late-time long-duration error is still a relatively high 57.3\%, but represents a significant improvement over the FFT MCA.  Additionally, the important resonant peaks are clearer with the ESP approach, especially the ones at 18kHz and 42kHz.

\begin{figure}
\centering
\includegraphics[width=3.25in]{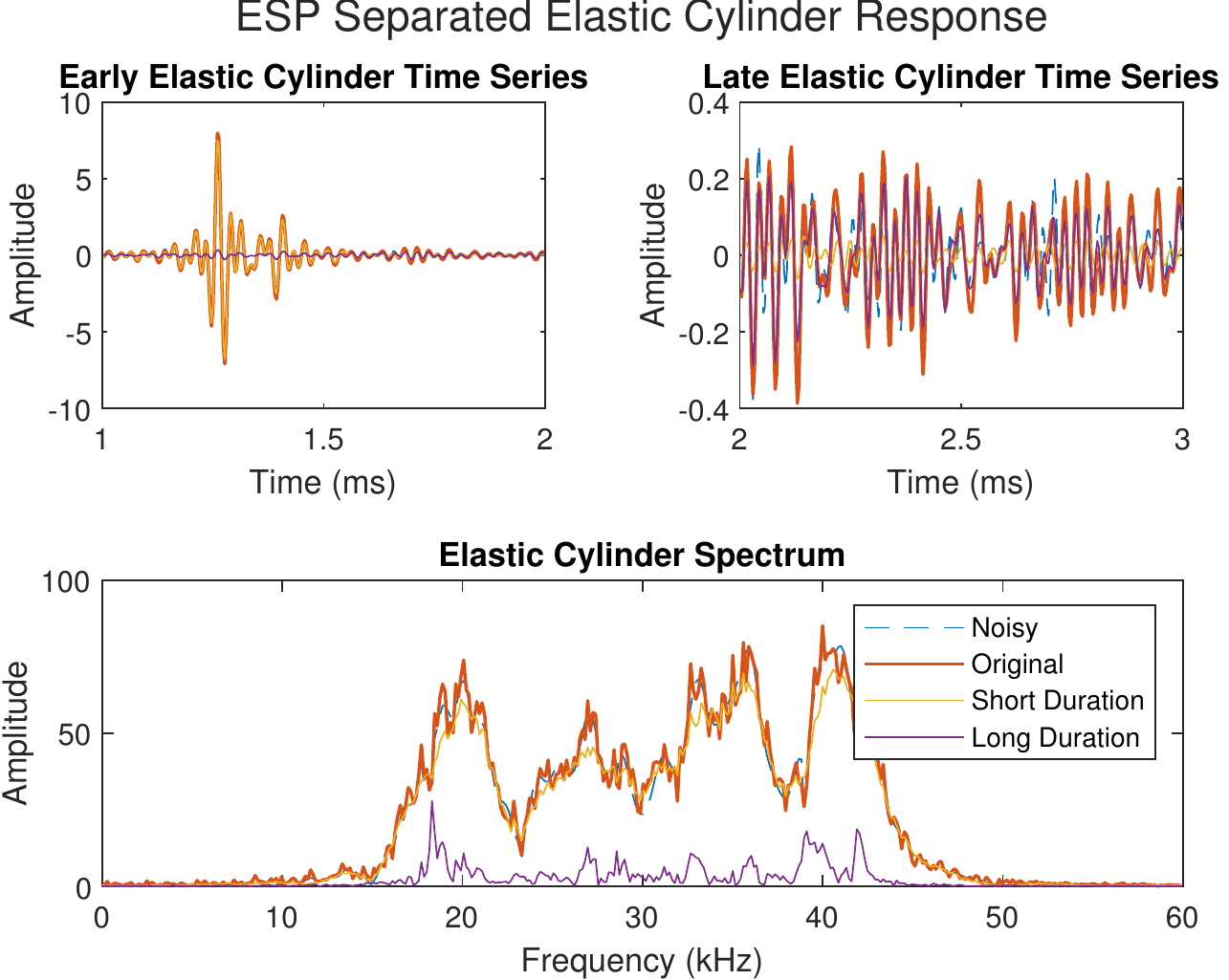}
\caption{(color online) Clean matched filtered LFM response, noisy matched filtered LFM response and the ESP separated short and long duration components for the elastic cylinder (top) and corresponding spectral power (bottom).  Noise was added at a 10dB SNR level.  The short-duration component has an early-time relative error of 11.1\% and the long-duration component has a late-time relative error of 57.3\%.}
\label{fig:stanton-noise-esp}
\end{figure}

 Table~\ref{table:stanton-noise} summarizes the results of this section. We see that both methods have a large late-time, long-duration component error.  This is because the added noise was much louder compared to the quiet ring down, which impacts the performance of BPD.  It is also worth noting that  one side effect of BPD is the overall reduction of signal power and most of the signal power lives in the early time.  This is an important contributor to the early-time short-duration error. ESP MCA is the best performer of the methods, with significantly lower error rates, particularly in the late-time.  This is not unexpected since ESP MCA was designed to fit the underlying signal while FFT MCA is signal agnostic.  
 Of course, these results are for a particular noise level and common choice of \(\lambda\).  The overall performance of an individual approach can be maximized by tuning \(\lambda\) as is done in the next section.

\begin{table}
\centering
\begin{tabular}{c | c c}
~ & \multicolumn{2}{c}{Relative Error} \\
Method & \(m_1\)  & \(m_2\) \\ \hline
FFT MCA  & 17.6\% & 75.4\% \\
ESP MCA  & {\bf 11.1\%} & {\bf 57.3\%}
\end{tabular}
\caption{Short-duration early-time relative error \(m_1\) and long-duration late-time relative error \(m_2\) for FFT MCA and ESP MCA for the noisy elastic cylinder.}
\label{table:stanton-noise}
\end{table}

\subsection{Noise Analysis}
\label{sec:stanton-noise}

In order to compare optimal MCA separation between the two techniques, experiments were performed using a range of noise levels and \(\lambda\)-values.  Gaussian noise was added to the LFM signal at SNR ranging from \(-10\)dB to 30dB with the same signal processing as described in Section~\ref{sec:stanton-lfm}.  For each SNR, 1000 noise realizations were instantiated and each was separated using FFT and ESP MCA with 1000 iterations of BPD for each of the following \(\lambda\)-values
\begin{equation}
\label{eq:lambdas}
\lambda_j = \lambda_{\max}10^{-3 + 0.25j}\ \text{for \(j = 0, \ldots, 11\).}
\end{equation}
The early-time relative error in the short-duration component and the late-time relative error in the long-duration component was computed for each combination of SNR, noise realization, and \(\lambda\)-value.  The means and standard deviations of those errors are plotted, for those \(\lambda_j\) which give the minimum error, in Figure~\ref{fig:stanton-noise-snr} as a function of SNR.

\begin{figure}
\centering
\includegraphics[width=3.25in]{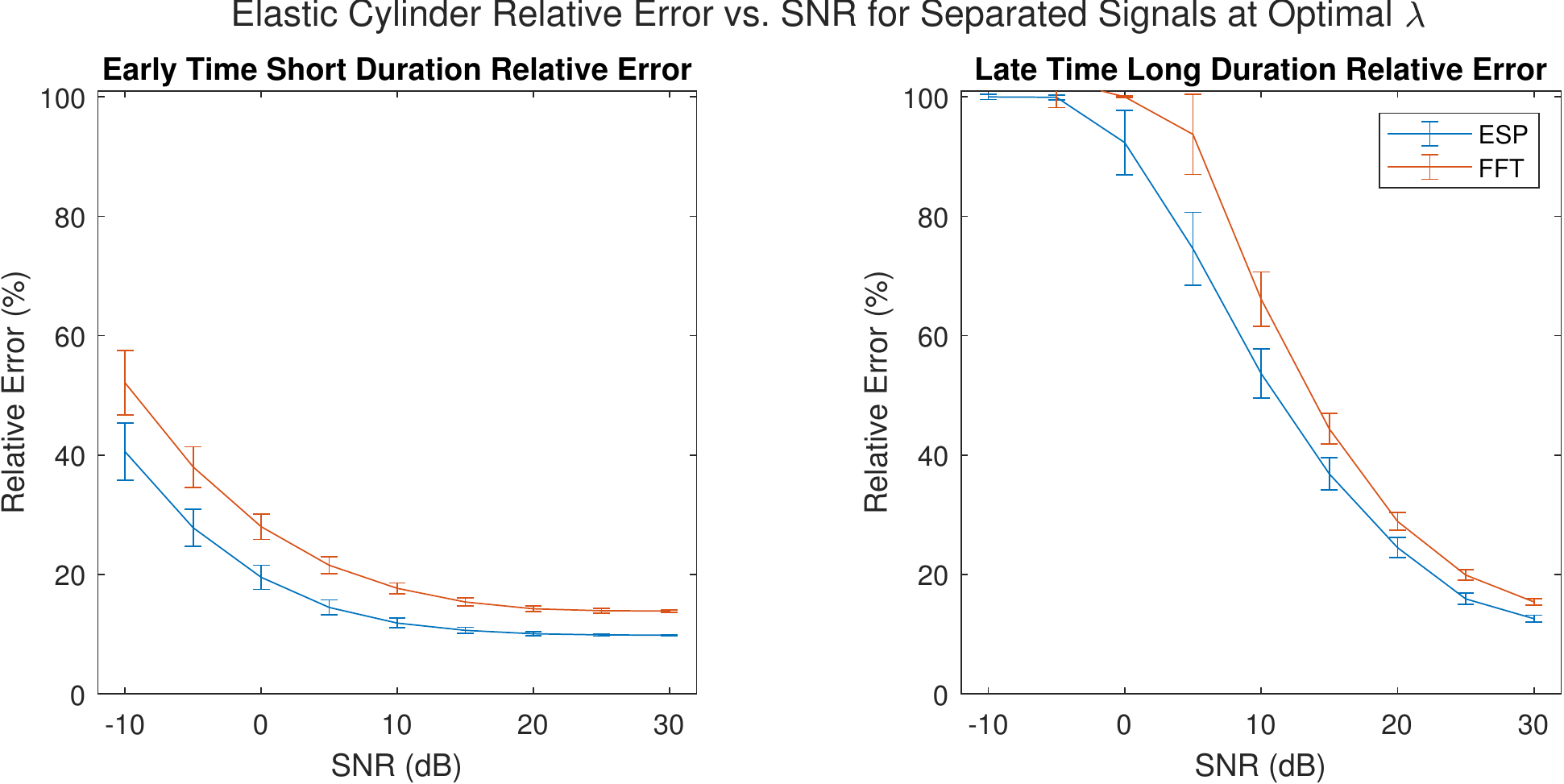}
\caption{(color online) Early-time short-duration component relative error (left) and late-time long-duration component relative error (right) for the elastic cylinder time series.  Data points represents mean error averaged over 1000 noise realizations and the bars represent the standard deviation.  The optimal \(\lambda\) was chosen from logarithmically selected \(\lambda\) ranging over \([0.001\lambda_{\max}, 0.5623\lambda_{\max}]\) using 1000 iterations of BPD.}
\label{fig:stanton-noise-snr}
\end{figure}

The ESP based separation performs better than the FFT based separation.  The short-duration errors in the left plot of Figure~\ref{fig:stanton-noise-snr} display typical behavior with larger relative errors and standard deviations for lower SNR.  None of the relative errors approach the 100\% mark, which is an indication that even at the highest noise levels there is some amount of signal in the separated component.  The short-duration component error is smallest for ESP MCA over all SNR levels.

For the long-duration component, at low SNR the relative error is either above 100\% with a large standard deviation, indicating that the long-duration component is picking up loud noise, or is very near 100\% with little deviation, indicating that the long-duration component is near zero.  In either case, any relative errors at or above 100\% should be considered degenerate.  We see more reasonable behavior above 5dB SNR with ESP MCA slightly outperforming FFT MCA.  Overall ESP MCA does a slightly better job of separation in the face of noise.  The FFT MCA performs reasonably well, considering it is signal agnostic, and is computationally less intensive.

Overall this section demonstrates that MCA can be utilized to separate short-duration and long-duration signal components, at least in the case of signals derived from analytic models.
While the separation was clearest in the case of a clean impulse response, reasonable results were obtained for the LFM response with noise as well.  For both scenarios, sharp narrow nulls in the signal spectrum were transformed into peaks in the spectrum of the late-time component, providing a more robust feature for detection and classification.  Successful separation was performed with SNR as low as 10dB (\(-10\)dB as measured against the late time component), a plausible noise regime given current sensors.

\section{AirSAS Signal Separation}
\label{sec:airsas}

In this section we will apply the MCA techniques presented in Section~\ref{sec:theory} to experimentally generated AirSAS time series \cite{airsas}.  Experimental AirSAS data was collected on two targets: an 8-inch long, 2-inch diameter copper pipe with 0.032-inch thick walls, and an 8-inch long, 2-inch diameter air-filled, hollow copper cylinder with 0.032-inch thick walls and end caps.  The targets were centered on a turntable and rotated in 1 degree increments relative to a transducer array consisting of loudspeaker tweeter (Peerless OX20SC02-04) and a microphone (GRAS 46AM), see Figure~\ref{fig:airsas}.  The tweeter transmits a 1ms LFM chirp from 30kHz to 10kHz and the microphone receives the signals backscattered from the target.  Motion, timing, signal generation and capture is controlled from a National Instruments data acquisition platform.  The recorded signals are matched filtered with the transmitted waveform. For this paper we only utilize the 3ms to 8ms portion of each time series.  This window captures the early-time response at all angles as well as much of the late-time response before the room clouds the data.

We apply FFT and ESP MCA to the resulting dataset for the 0.032-inch hollow copper cylinder object in Section~\ref{sec:airsas-ch} with a noise analysis on the same dataset in Section~\ref{sec:airsas-noise}. We finish with an application of FFT and ESP MCA to the more complicated time series collected from the 0.032-inch copper pipe in Section~\ref{sec:airsas-cp}.

\begin{figure}
\centering
\includegraphics[width=3in, angle=-90]{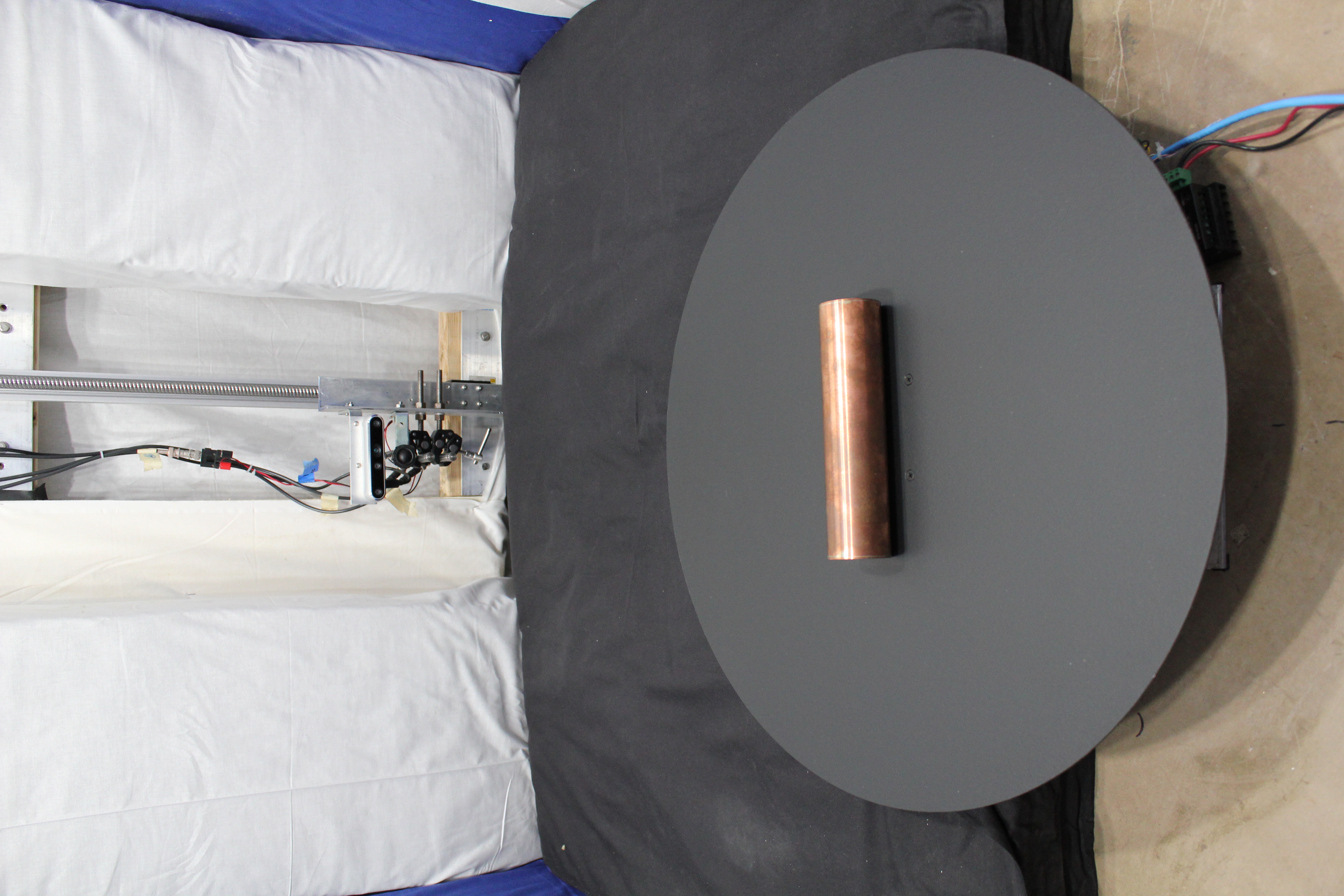}
\caption{(color online) The AirSAS experimental setup was used to collect data on 8-inch long, 2-inch diameter targets. A loudspeaker ensonifies the objects with an LFM chirp, and the backscattered echoes are recorded by a microphone.  A turntable rotates the targets relative to the transducers between successive pings.}
\label{fig:airsas}
\end{figure}

Despite the fact that the AirSAS data is much more complex than the analytic signal, with multiple returns arriving at different times, for this section we will continue use the performance metrics from \eqref{eq:metrics}.  However for the AirSAS cylinder data we use an early time-interval \(I_1\) from 4ms to 6ms and a late-time interval \(I_2\) from 6ms to 8ms.  Since each AirSAS object scan includes 360 different time series, we will report the relative error averaged over all aspect angles, which introduces some variation into the error.  While we will still view these metrics as a measure of separation performance, the fact that we expect there to be late-time short-duration energy (particularly in Section~\ref{sec:airsas-cp}) means that even in the case of perfect separation we would not expect either \(m_1\) or \(m_2\) to be zero. More broadly these metrics provide only a rough indication of overall performance.

\subsection{0.032-inch Hollow Copper Cylinder}
\label{sec:airsas-ch}

The first dataset we will consider are the AirSAS time series collected from the 0.032-inch hollow copper cylinder.  In many respects, these time series are a reasonable analogue to the Stanton model used in Section~\ref{sec:stanton}, since at most aspect angles there is a single bright initial return potentially followed by a long-duration low power component.  Notably thin walled pipes support a wider class of non-rigid phenomenon than Stanton's model.  Various representations of the hollow copper cylinder experimental data are shown in Figure~\ref{fig:airsas-chbf}.  The top left subplot is a logarithmically scaled color plot of the time series amplitude. The bottom left subplot shows the associated normalized target strength, which is the spectra of each time series normalized across all aspect angles.  The top right subplot is a logarithmically scaled color plot of the Polar Format Algorithm (PFA) generated image magnitude \cite{doerry-pfa}.  The bottom right subplot shows the object's \(k\)-space representation, which is the magnitude of the two-dimensional Fourier Transform of the complex PFA image.  The long-duration signal is clearly present in the time series representation, in bands from \(-10\) degrees to 90 degrees and 180 degrees to 280 degrees.  This late time energy is also apparent in the PFA image.  Not readily apparent in either of the spectral representations is a faint signature corresponding to this late-time energy.

\begin{figure*}[h]
\centering
\includegraphics[width=4in]{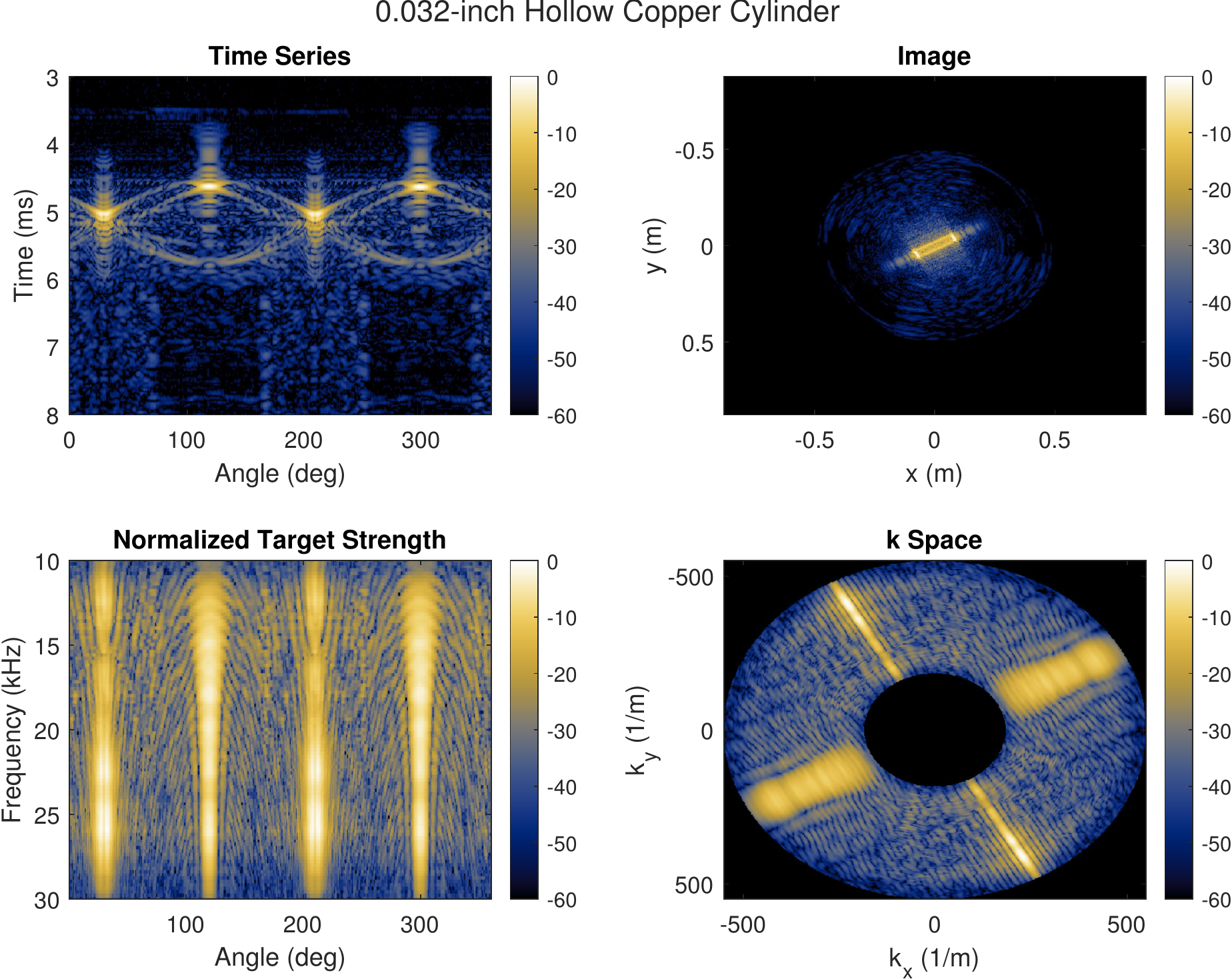}
\caption{(color online) Time series (top left), PFA image (top right), normalized target strength (bottom left), and \(k\)-space (bottom right) for the 0.032-inch hollow copper cylinder. All plots logarithmically scaled.}
\label{fig:airsas-chbf}
\end{figure*}

MCA is broadly compatible with the signal processing and image reconstruction algorithms used with the AirSAS data.  For this section, we will apply MCA to the matched filtered AirSAS time series individually, splitting each into  short-duration and long-duration components.  We then apply PFA to the separated time series to reconstruct a pair of images, one corresponding to the short-duration components and the other to the long-duration components.  We produce normalized target strength representations corresponding to the short-duration and long-duration components as well.

\subsubsection*{FFT MCA}

To begin, we will apply FFT MCA using 1000 iterations of BP to the 0.032-inch hollow copper cylinder as described above.  Since we are utilizing BP, the separated time series as well as the corresponding PFA images will add up exactly to the original dataset from Figure~\ref{fig:airsas-chbf}.  After image formation, the PFA images associated to the separated short-duration and long-duration components are shown in Figure~\ref{fig:airsas-fft-ch}, along with their normalized target strength representations, on a pairwise common color scale.  The separation looks fairly clean in the PFA image, with the extended ringing response from the cylinder principally in the long-duration image while the brighter geometric scattering response appears in the short-duration image.  There does appear to be some bleed-through of the object into the long-duration image.  The average short-duration early-time error is 44.9\% while the average long-duration late-time error is 23.7\%.  One particularly interesting set of features are the hyperbolic signatures present at 20 degrees and 210 degrees in the long-duration normalized target strength plot, since these signatures were masked by the much brighter short-duration response in the bottom-left plot of Figure~\ref{fig:airsas-chbf}.

\begin{figure*}[h]
\centering
\includegraphics[width=4in]{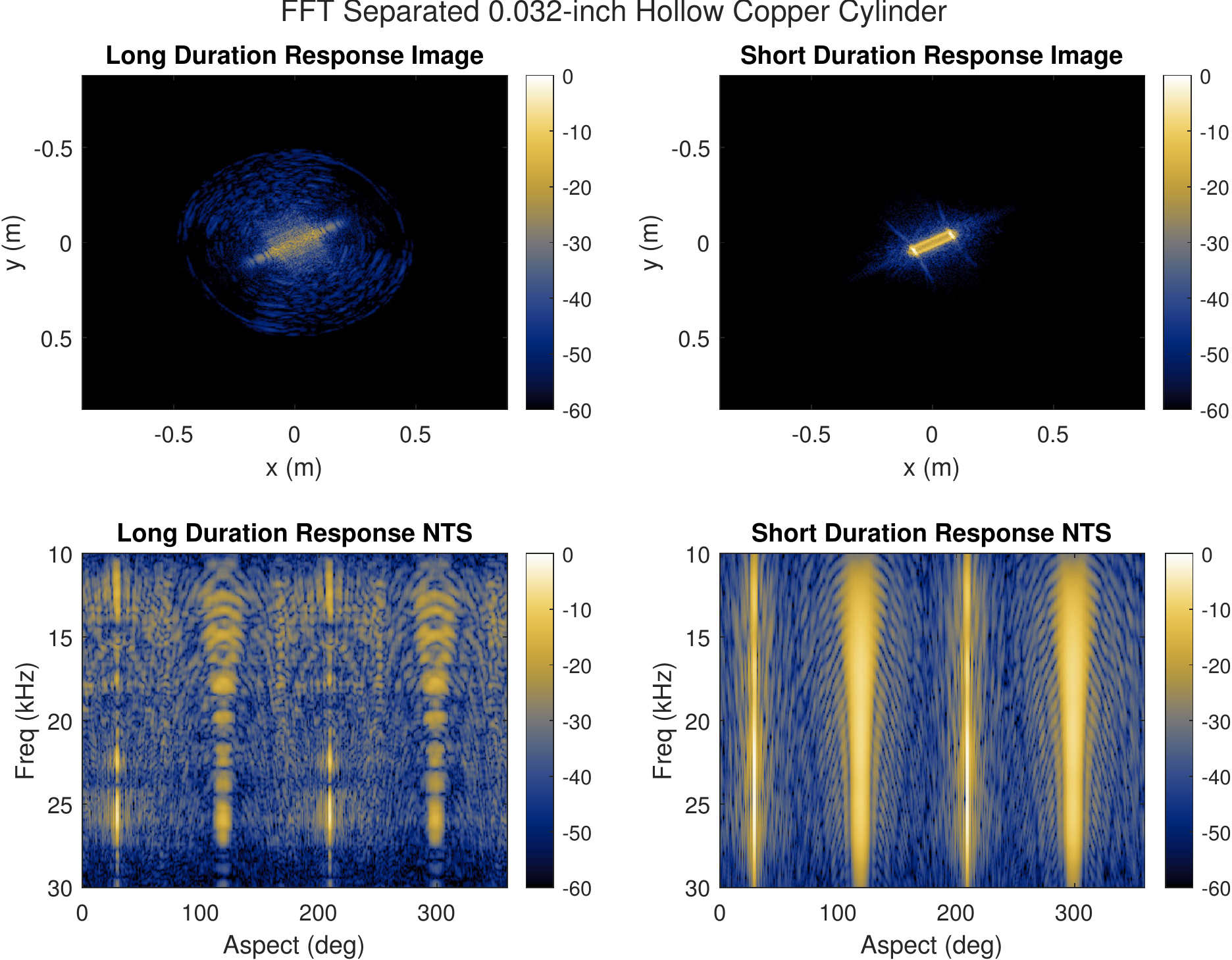}
\caption{(color online) PFA image (top) and normalized target strength (bottom) for the FFT separated long-duration component (left) and short-duration component (right) of the hollow copper cylinder object.  The plots share a pairwise common color scale.  The short-duration component has an average early-time relative error of 44.9\% and the long-duration component has an average late-time relative error of 23.7\%. }
\label{fig:airsas-fft-ch}
\end{figure*}

\subsubsection*{ESP MCA}

For ESP MCA we will use the same process as above, applying BP to individual time series and producing a pair of PFA images associated to each component.  We continue to use rectangular windows and decaying exponentials as our envelopes, with the same window lengths and time constants as \eqref{eq:esp-param}.  Heuristic experimentation showed these parameters produced reasonable results, although we will see that specific performance characteristics can be attained by using shorter windows and larger time constants. Using BP with 1000 iterations produces the PFA images shown in Figure~\ref{fig:airsas-esp-ch}.  The separation is quite effective with the bright initial scattering almost completely contained within the short-duration response with an average early-time short-duration error of 19.0\%.  The long-duration component has most, but not all, of the late-time response and has an average error of 49.0\%.  Interestingly, there is some late-time energy in the short-duration image at the acoustic coupling angles that was not present in the FFT MCA.  This energy does appear to take the form of late arriving wave packets and one theory is that these discrete returns are late arriving pulses associated to surface waves propagating on the cylinder. If this is indeed the case, then it is more appropriate for them to be part of the short-duration image. The fact that they are not present in the long-duration image negatively impacts the long-duration late-time relative error, which is consistent with our goal of separating the signal into early-time and late-time components.

\begin{figure*}[h]
\centering
\includegraphics[width=4in]{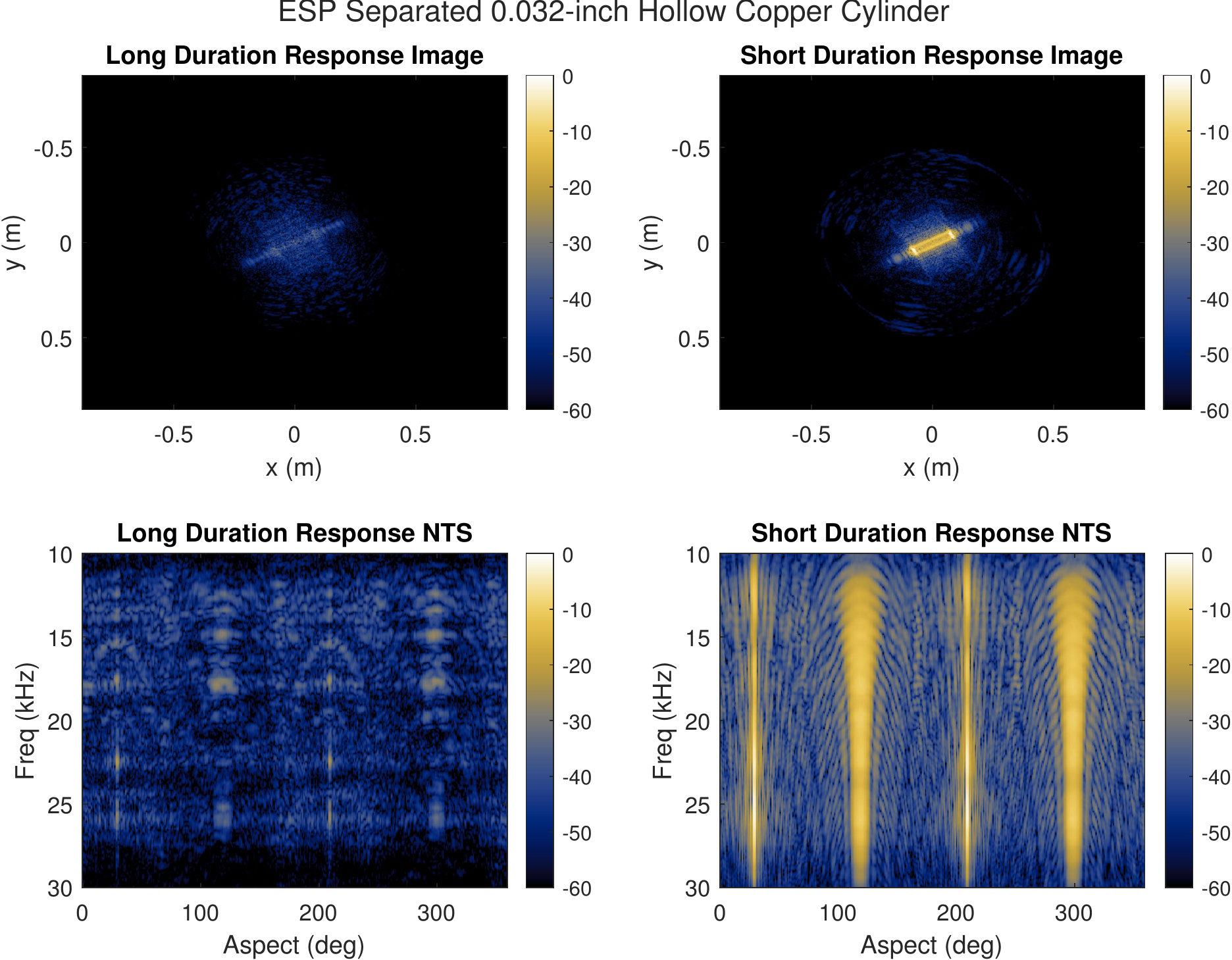}
\caption{(color online) PFA image (top) and normalized target strength (bottom) for the ESP separated long-duration component (left) and short-duration component (right) of the hollow copper cylinder object.  The plots share a pairwise common color scale. The short-duration component has an average early-time relative error of 19.0\% and the-long duration component has an average late-time relative error of 49.0\%.}
\label{fig:airsas-esp-ch}
\end{figure*}

Overall, compared to the FFT MCA separation we have sacrificed accuracy in the long-duration late-time for increased accuracy in the short-duration early-time, which we see in Table~\ref{table:airsas-ch}.  As we will demonstrate in Section~\ref{sec:airsas-cp} the ESP frame approach is flexible so this separation could be further tuned by changing the envelope parameters.  (Recall that using a one-hot and constant envelope will largely reproduce the FFT result.)  Moreover, there is evidence that not all of the short-duration components are early-time, which impacts the reliability of the early-time/late-time error metrics.  Visually the separation appears best with the ESP MCA approach as the hyperbolic late-time features are clearer and at a higher relative power.

\begin{table}
\centering
\begin{tabular}{c | c c }
~ & \multicolumn{2}{c}{Relative Error} \\
Method & \(m_1\)  & \(m_2\) \\ \hline
FFT MCA & 44.9\% & {\bf 23.7\%}  \\
ESP MCA & {\bf 19.0\%} & 49.0\% \\
\end{tabular}
\caption{Short-duration early-time relative error \(m_1\) and long-duration late-time relative error \(m_2\) for FFT MCA and ESP MCA for the copper hollow cylinder.}
\label{table:airsas-ch}
\end{table}

\subsection{Noise Analysis}
\label{sec:airsas-noise}

In order to understand FFT MCA and ESP MCA separation behavior over a broader range of noise levels, the analysis of Section~\ref{sec:stanton-noise} has been reproduced using AirSAS data.  Specifically Gaussian noise was added to the time series at each aspect angle with SNR ranging from 10dB to 40dB, all measured against a common reference power.  The individual time series were then separated using FFT MCA and ESP MCA using BPD with 1000 iterations and \(\lambda\)-values given by \eqref{eq:lambdas}. The early-time relative error in the short-duration component and the late-time relative error in the long-duration component was computed for each combination of noise level, and \(\lambda\)-value.  We then computed the mean and standard deviation of those errors over all aspect angles.  The mean and standard deviation are plotted in Figure~\ref{fig:airsas-noise-error-snr} as a function of SNR for those \(\lambda_j\) which give the minimum error.

\begin{figure}
\centering
\includegraphics[width=3.25in]{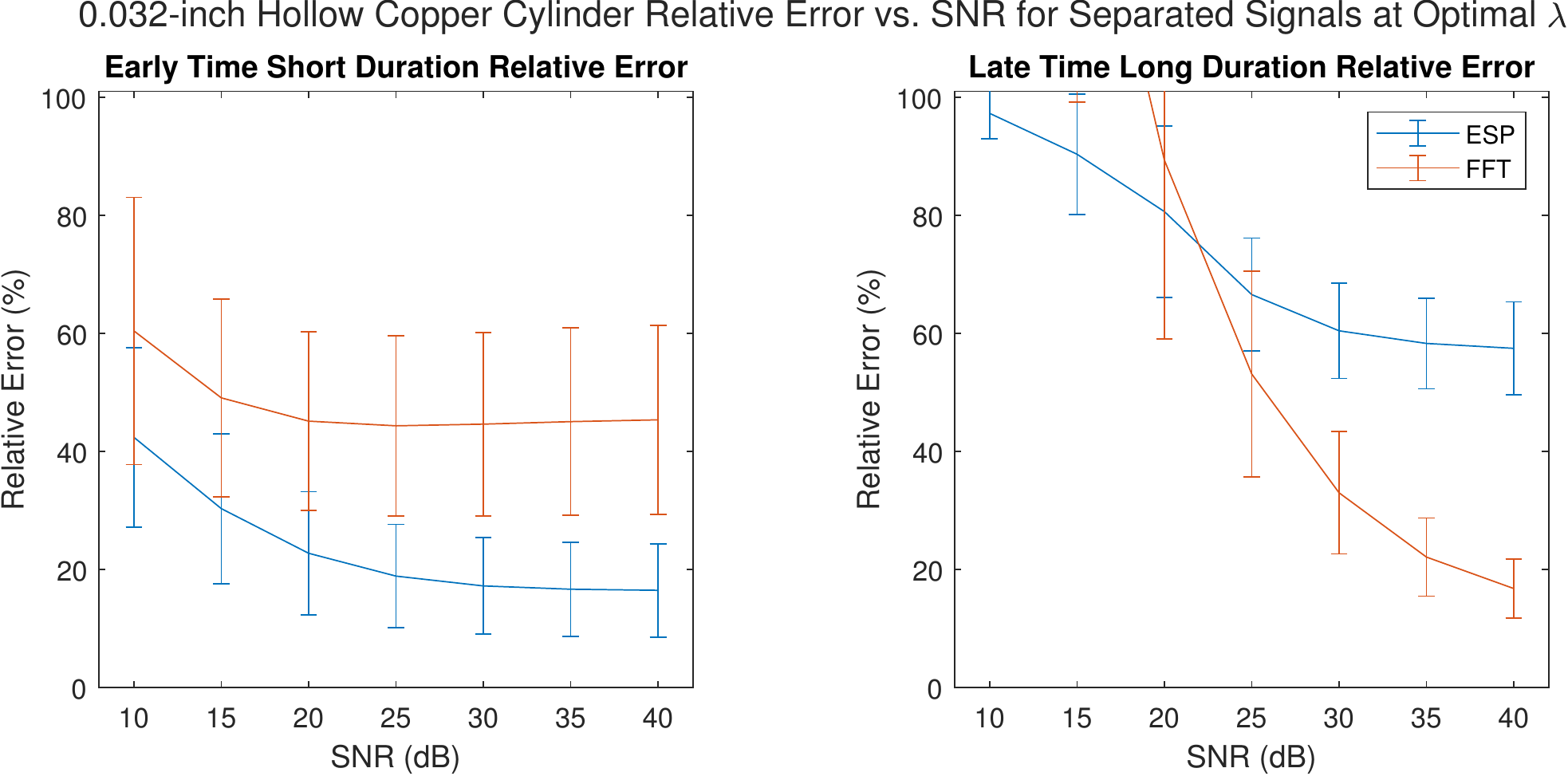}
\caption{(color online) Early-time short-duration component relative error (left) and late-time long-duration component relative error (right) for copper hollow cylinder AirSAS data.  Data points represents mean error averaged over aspect angle and the bars represent the standard deviation.  The optimal \(\lambda\) was chosen from logarithmically selected \(\lambda\) ranging over \([0.0001\lambda_{\max}, 0.5623\lambda_{\max}]\) using 1000 iterations of BPD.}
\label{fig:airsas-noise-error-snr}
\end{figure}

The results in Figure~\ref{fig:airsas-noise-error-snr} are consistent with the single sample results in the previous section.  First, observe that the standard deviations do not decrease as a function of SNR as they do in Section~\ref{sec:stanton-noise}.  This is because we are averaging over all 360 aspect angles, not averaging over different noise realizations on the same time series.  In terms of mean error, ESP MCA has a consistently lower error than FFT MCA for the early-time short-duration component.  Both errors tend to decrease as a function of SNR, as expected.

For the long-duration component the late-time relative error is near or above 100\% for SNR less than \(\approx 15\)dB for both methods, indicating a failure of separation. There is evidence that ESP MCA starts producing viable separations at 20dB, while both methods produce viable separations by 25dB.  However, the FFT MCA separation produces a lower late-time long-duration relative error for SNR greater than \(\approx 22\)dB.

More broadly we have demonstrated the ability to separate SAS imagery into distinct morphological components using both FFT MCA and ESP MCA.  The FFT MCA approach is superior at producing separations with lower late-time error while ESP MCA has consistently lower early-time error.  Importantly both MCA methods were designed to separate short-duration and long-duration components, rather than early-time/late-time components, and will associate the late-time short-duration energy in the AirSAS data with geometric scattering.  An open question is if differences in the morphology between early-time and short-duration late-time energy could be used to further the overall goal of early-time/late-time separation.

\subsection{0.032-inch Copper Pipe}
\label{sec:airsas-cp}

The second data set we will consider are the AirSAS time series collected from the 0.032-inch copper pipe.  This dataset is significantly more complex than the 0.032-inch hollow copper cylinder dataset, with obvious short-duration late-time energy present in the time-series.  This will exercise the MCA approach to signal separation by demonstrating separation of a long duration signal superimposed on a sequence of repeated short duration signals, but will further decrease the utility of our performance metrics.  Figure~\ref{fig:airsas-chbf} shows logarithmically scaled color plots of the 0.032-inch copper pipe time series data, the associated normalized target strength, as well as the PFA image and its associated \(k\)-space representation.  The aforementioned late-time short-duration energy is present in discrete ``rings'' around the object from \(-10\) degrees to 90 degrees and 180 degrees to 280 degrees.  We wish to understand how our MCA tools respond to this energy.

\begin{figure*}[h]
\centering
\includegraphics[width=4in]{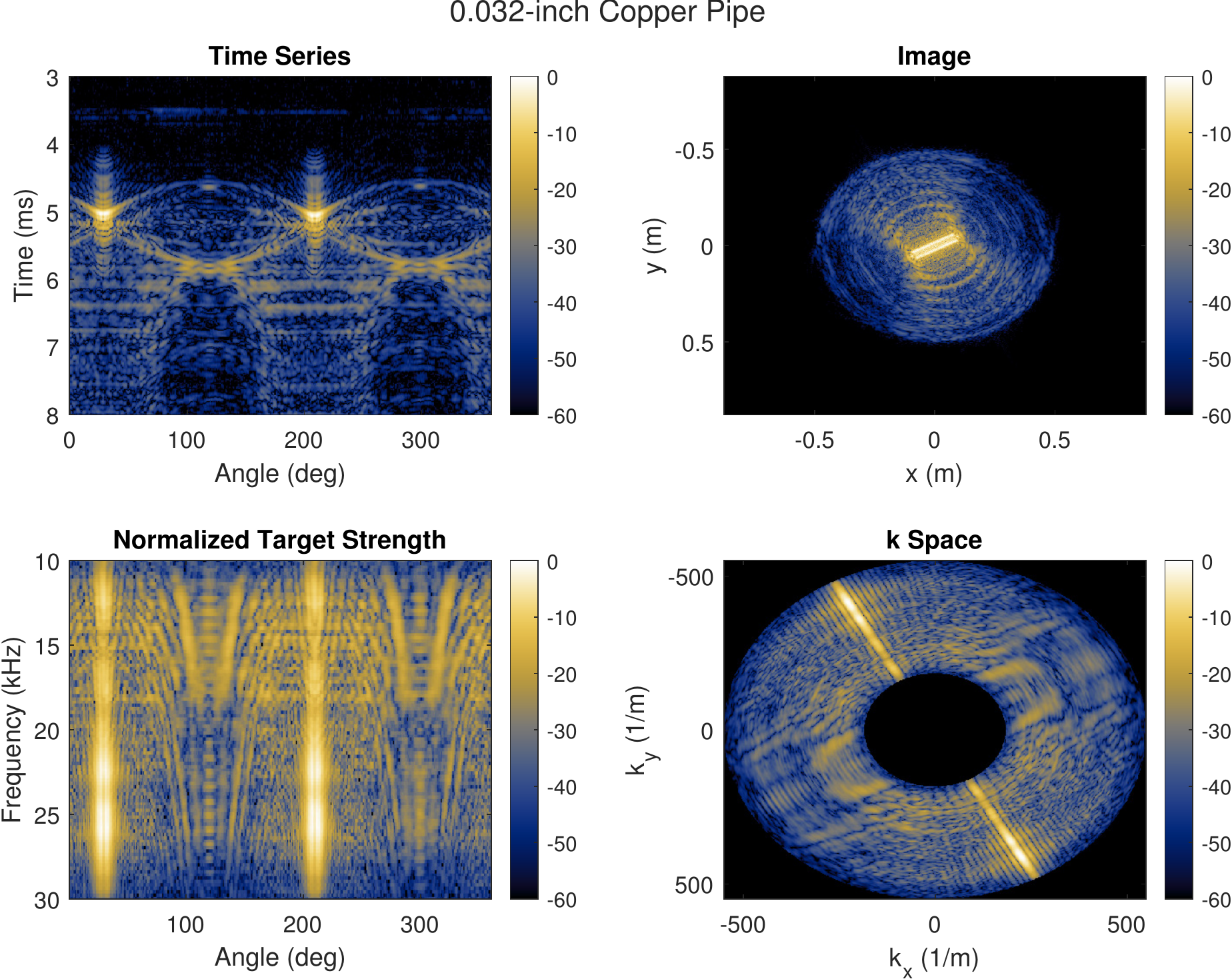}
\caption{(color online) Time series (top left), PFA image (top right), normalized target strength (bottom left) and \(k\)-space (bottom right) 0.032-inch copper pipe. All plots logarithmically scaled.}
\label{fig:airsas-cpbf}
\end{figure*}

\subsubsection*{FFT MCA}

Since FFT MCA is signal agnostic we apply it to the 0.032-inch copper pipe just as in Section~\ref{sec:airsas-ch}.  Using 1000 iterations of FFT MCA BP we produce the separated PFA images shown in Figure~\ref{fig:airsas-fft-cp}.  It is immediately apparent that most of the late-time energy, including the shorter duration late-time ``rings,'' are contained in the long duration image.  As a result the average late-time, error for the long-duration component is a relatively low 17.7\%, at the cost of a higher 64.1\% error for the early-time short-duration component.  It is slightly unexpected that these apparently short duration signals can be more sparsely represented in the frequency domain; however, analysis of the spectrum shows that while this late time energy is apparently time limited it is nevertheless not particularly broad band.  This is due to the fact that these late arriving wave packets are shaped reflections of the LFM used to ensonify the object.

\begin{figure*}[h]
\centering
\includegraphics[width=4in]{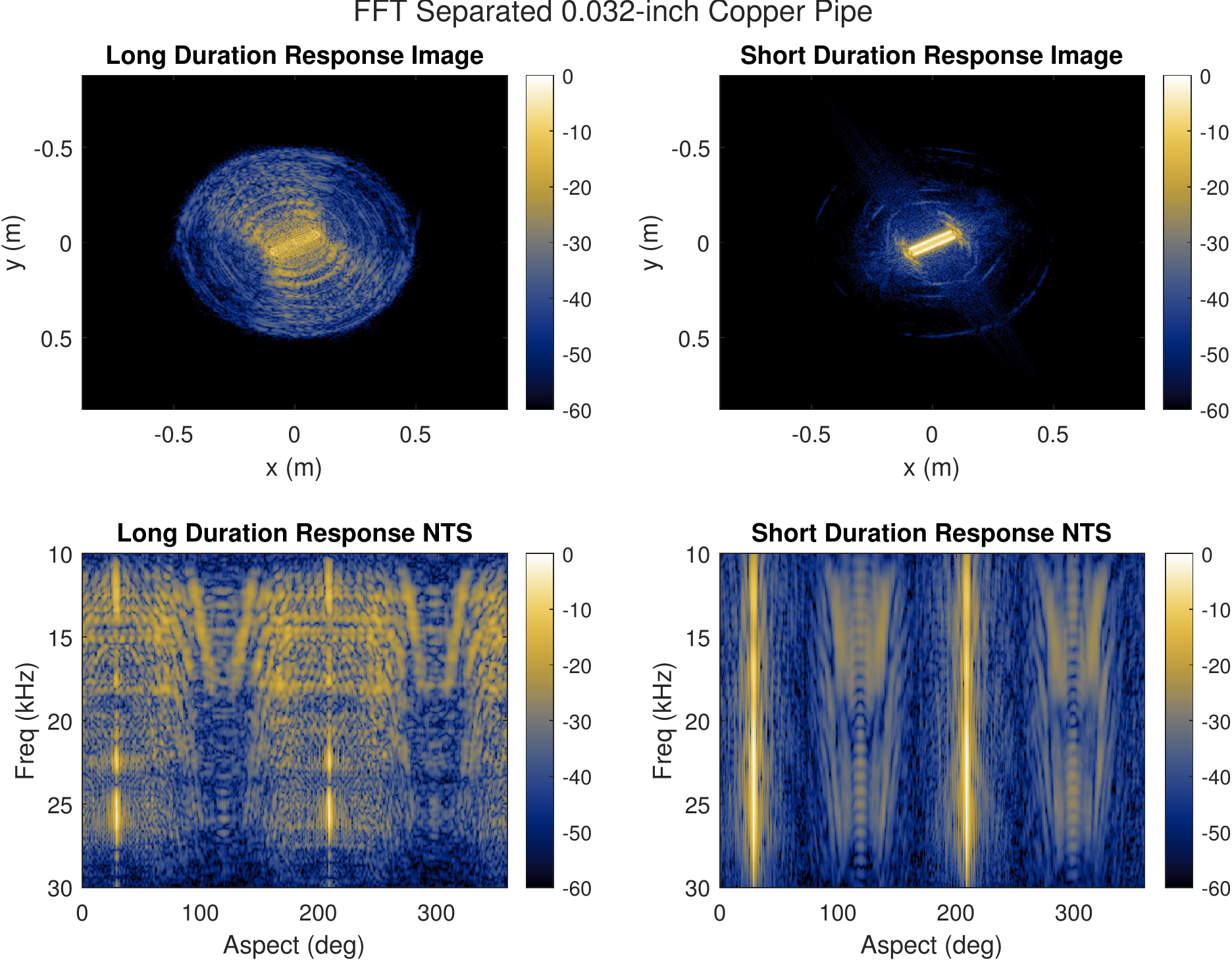}
\caption{(color online) PFA image (top) and normalized target strength (bottom) for the FFT separated long-duration component (left) and short-duration component (right) of the copper pipe object. The plots share a pairwise common color scale.  The short-duration component has an average early-time relative error of 64.1\% and the long-duration component has an average late-time relative error of 17.7\%. }
\label{fig:airsas-fft-cp}
\end{figure*}

\subsubsection*{ESP MCA}

Next we will separate the copper pipe time series using the same ESP MCA parameters as Section~\ref{sec:airsas-ch}.  The short and long-duration PFA images resulting from 1000 iterations of BP are shown in  Figure~\ref{fig:airsas-esp-cp}.  We see that unlike the FFT MCA much of the power of the late-time short-duration wavepackets has been placed in the short-duration PFA image.  The separation is a bit muddled overall, although there is clear distributed late-time energy in the long-duration plot over the expected range of angles.  For ESP MCA, the long-duration error is worse than the FFT case, with an average late-time error of 62.5\%, but the average short-duration early-time error is a better 27.4\%.  Also notable is that the normalized target strength plots in Figures~\ref{fig:airsas-fft-cp} and \ref{fig:airsas-esp-cp} emphasize different features. The ``V'' shaped signatures between 10kHz and 20kHz around 130 degrees and 300 degrees seem to move from the short-duration component to the long-duration component.

\begin{figure*}[h]
\centering
\includegraphics[width=4in]{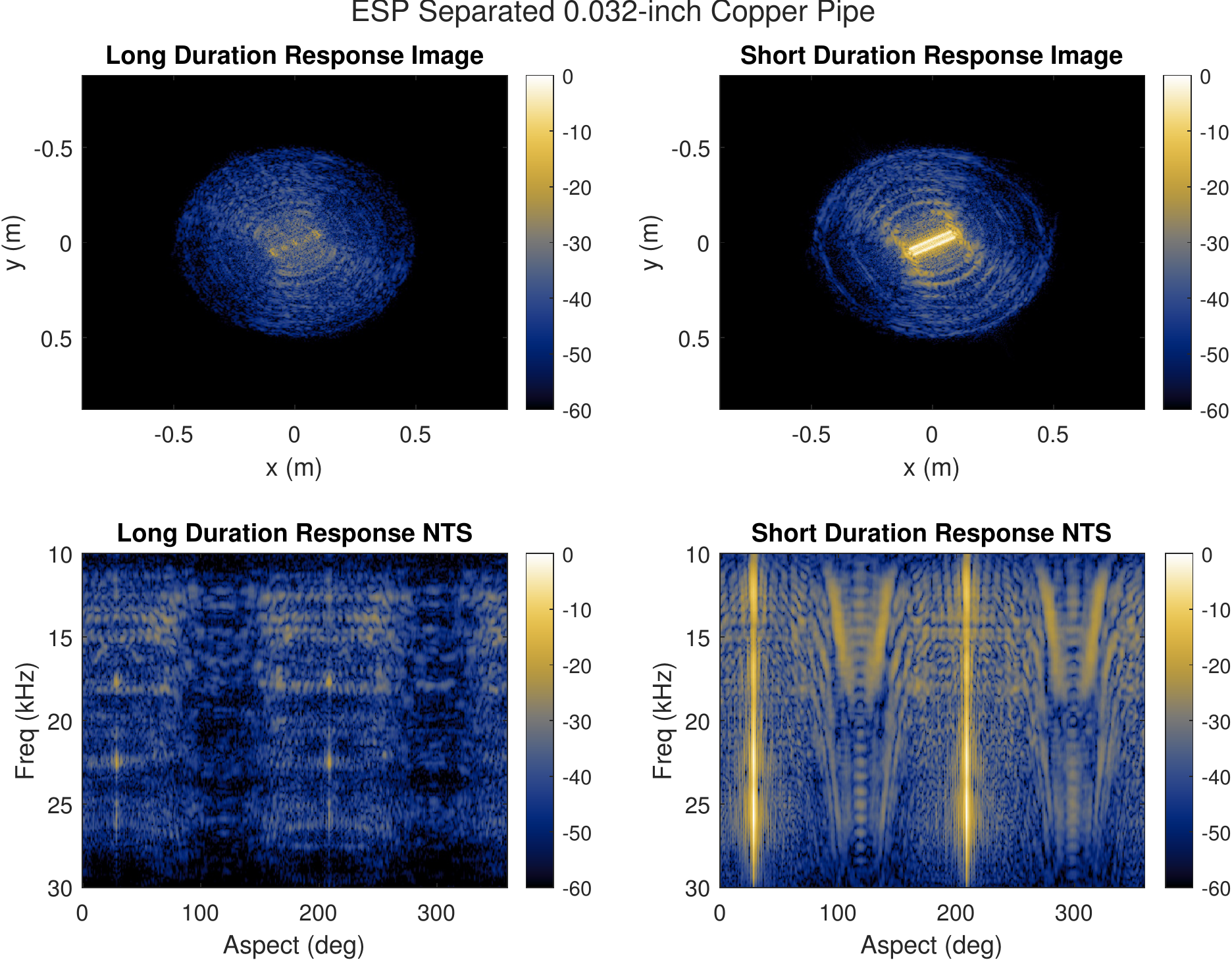}
\caption{(color online) PFA image (top) and normalized target strength (bottom) for the ESP separated long-duration component (left) and short-duration component (right) of the copper pipe object.  The plots share a pairwise common color scale.  The short-duration component has an average early-time relative error of 27.4\% and the long-duration component has an average late-time relative error of 62.5\%. }
\label{fig:airsas-esp-cp}
\end{figure*}

Traditionally, one would tune the MCA \(\lambda_i\) parameters in order to move signal energy between the separated components to achieve some desired result.  While this is often a useful practical step, the determination of \(\lambda_i\) is not obviously connected to the underlying signal characteristics.  One of the benefits of the ESP MCA approach is that ESP windows can be tuned based on properties of the signal in question.  For example, if it is desired that the long duration component include all (or more) of the late-time energy to improve early-time/late-time separation, then that can be accomplished by shortening the windows available to the short-duration component and including more quickly decaying exponentials in the long-duration representation.  This makes it easier for the exponential envelopes to sparsely represent shorter signals and more expensive for the rectangular windows to compete for the same short signals.  In this case, we shorten our rectangular window lengths considerably and add an additional time constant of 1ms, resulting in:
\begin{align*}
T_l &= 0.01, 0.05\text{ms}, \\
\tau_l &= 1.00, 1.78, 3.16, 5.62, 10.00, 17.78, 31.62\text{ms}.
\end{align*}
Performing 1000 iterations of ESP MCA BP with these parameters results in the separation shown in Figure~\ref{fig:airsas-alt-esp-cp}.

\begin{figure*}[h]
\centering
\includegraphics[width=4in]{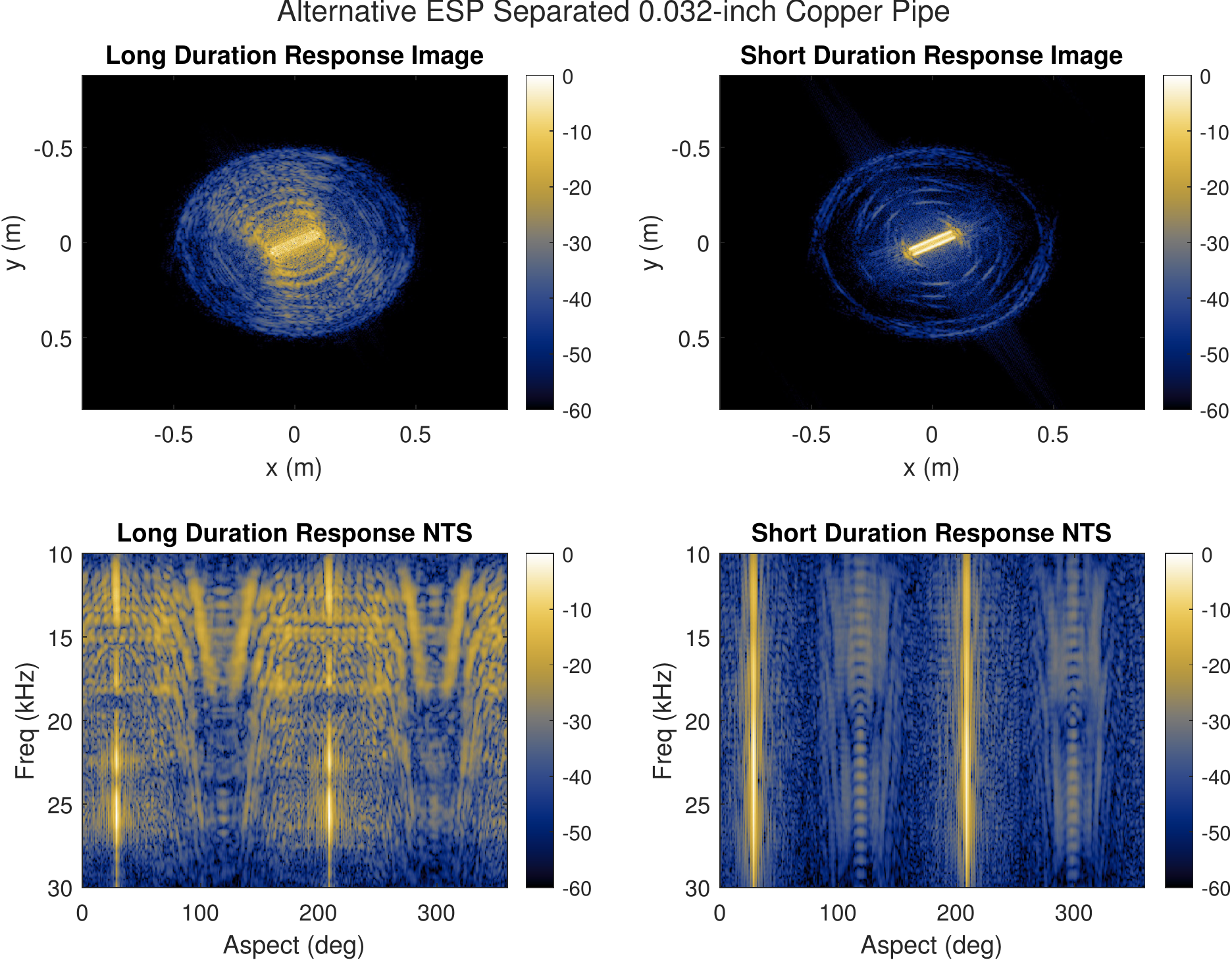}
\caption{(color online) PFA image (top) and normalized target strength (bottom) for the ESP separated long-duration component (left) and short-duration component (right) of the copper pipe object using an alternative set of envelopes.  The plots share a pairwise common color scale.  The short-duration component has an average early-time relative error of 74.6\% (3 dB) and the long-duration component has an average late-time relative error of 15.97\% (16 dB). }
\label{fig:airsas-alt-esp-cp}
\end{figure*}

In this case we get something that looks more like the FFT separation, with most of the late-time power (including the short-duration components) in the long-duration image.  There is more latetime energy in the short-duration image for this alternative ESP MCA than FFT MCA and the average short-duration early-time error is larger at 74.6\%.  The long-duration late-time average error 15.9\%, an improvement over FFT MCA.  Table~\ref{table:airsas-cp} summarizes the performance metrics for the copper pipe, with the alternative ESP MCA producing the best early time image and the ESP MCA producing the best late time image.  In this case, the table does not tell the whole story.  The overall performance of each method depends on whether the motivating need for MCA signal separation would benefit from late-time short-duration energy grouped with the initial scattering return or with the more diffuse late-time return.

\begin{table}
\centering
\begin{tabular}{c | c c }
Method & \(m_1\)  & \(m_2\) \\ \hline
FFT MCA             & 64.1\% & 17.7\%  \\
ESP MCA             & {\bf 27.4\%} & 62.5\% \\
Alternative ESP MCA & 74.6\% & {\bf 15.9\%}
\end{tabular}
\caption{Short-duration early-time relative error \(m_1\) and long-duration late-time relative error \(m_2\) for FFT MCA and ESP MCA for the copper pipe.}
\label{table:airsas-cp}
\end{table}

This section demonstrates some of the complexities of applying MCA to experimental data.  Performance of MCA representations will depend on the signals in question and determining the correct separation parameters is not always obvious.  In the case of ESP MCA, adjusting the envelope parameters can be done in a more principled fashion than adjusting the \(\lambda_i\)-parameters.  In either case, finding a way to allow the MCA representations to be data informed, rather than completely model driven, is a potential topic for further research.

\section{Conclusions}
\label{sec:discussion}

Motivated by the problem of separating the early-time and late-time returns from the acoustic response of an elastic object, we have presented a pair of MCA techniques which can successfully separate the short-duration and long-duration components without the need for a reference signal or time gates.  This isolates the late-time returns, with the geometric scattering response generally present in the short-duration component and high-\(Q\) resonances present in the long-duration components.  Successful separation was achieved on both clean and noisy analytic model data as well as experimentally collected in-air sonar time series.

The FFT MCA approach is signal agnostic and does a reasonable job of signal separation, performing very similarly to ESP MCA in the clean Stanton model case (see Table~\ref{table:stanton}).  FFT MCA also has the benefit of being extremely fast, but is rigid and cannot be tuned to fit a particular signal outside of the traditional \(\lambda_i\) parameters.  The ESP MCA application had the best metric scores in analytic time series analysis.  ESP MCA has a flexible signal model which can be tuned to support a wide variety of signals as demonstrated in Section~\ref{sec:airsas-cp}.  It does take orders of magnitude longer to run than FFT MCA, however.

When applied to noisy signals we found in Section~\ref{sec:stanton-noise} that MCA can turn sharp nulls in the spectrum of a time series, which are easily filled in by Gaussian noise, into peaks in the spectrum of the long-duration component, which are easier to identify in the presence of noise.  This is most obvious in the Stanton's model impulse response (see Figures~\ref{fig:stanton-fft} and \ref{fig:stanton-esp}) but can also be seen in the ESP MCA separation of the noisy Stanton's model LFM response (see Figure~\ref{fig:stanton-noise-esp}).  We also find evidence of this behavior in experimentally collected AirSAS data with the hyperbolic signatures that appear in the long-duration normalized target strength plots for the 0.032-inch hollow copper cylinder dataset (see Figure~\ref{fig:airsas-esp-ch}) but are obscured by the stronger short-duration geometric response in the unseparated data.  The ability to utilize features derived from separated components is an ongoing topic of research.

Additionally, Section~\ref{sec:airsas} demonstrates compatibility of MCA with SAS image reconstruction.  The ESP and FFT MCA techniques produced short-duration and long-duration PFA images which split the initial loud response of thin-walled cylindrical objects from the diffuse energy produced by long-duration ringing.  This was done without time gating and in the presence of both experimental noise and overlapping returns with varying start times.  Notably we have presented examples where long-duration late-time energy is separated from superimposed short-duration late-time energy (see Figure~\ref{fig:airsas-esp-cp}).  While the performance metric results were less clear cut for the experimental data, ESP MCA was capable of performing similar to FFT MCA (see Table~\ref{table:airsas-cp}) while being significantly more flexible in its signal model.  While both MCA approaches were designed to separate short-duration/long-duration components, the ultimate goal is to separate early-time/late-time components in order to preserve the assumptions of the image formation model.  It is an open question if the spectral characteristics of early-time versus short-duration late-time responses could be used in an MCA context.  The interplay between MCA and multi-path arrivals is an additional, related, topic currently being studied.

Moving forward there are quite a few potential applications of MCA to acoustic signals.  The separation process could be used in the formation of SAS imagery to either reduce late-arriving energy, allowing for a sharper representation of the object, or to highlight late-time ringing energy, identifying objects with elastic behaviors from those without.  Additionally, the spectral peaks resulting from high-\(Q\) modes in the long-duration component could be used as features for classification.  More broadly, MCA could be used to at least partially separate any components of a signal which feature sparse representation, such as overlapping acoustic returns from two pings with significantly different spectra, and as such FFT and ESP MCA provide flexible tools for tackling a fundamental acoustics challenge.

\begin{acknowledgments}
This work was sponsored in part by the Department of the Navy, Office of Naval Research under ONR award numbers N00014-18-1-2820 and N00014-19-1-2221.
\end{acknowledgments}

\appendix
\section{ESP Frames}
\label{apx:esp}

Consider the complex finite-dimensional Hilbert space \(\mathbb{C}^N\). All norms ($\|\cdot\|$) are computed in the $\ell_2$ sense unless otherwise stated. A {\em tight frame} is a collection of vectors \(\{\ba_i\}_{i=0}^{K-1}\) in \(\mathbb{C}^N\) and $\alpha > 0$ such that
\begin{equation}
\label{eq:tight-identity}
\|\bw\|^2 = \alpha\sum_{i=0}^{K-1} |\langle \bw, \ba_i\rangle|^2 \ \text{for all \(\bw\in \mathbb{C}^N\).}
\end{equation}
A {\em Parseval frame} is a tight frame with \(\alpha = 1\).  We define the class of {\em Enveloped Sinusoid Parseval Frames} by applying enveloping functions to the Discrete Fourier Transform (DFT) basis as in the following
\begin{theorem}
\label{thm:esp}
Given a set of nonzero \(N\)-dimensional vectors \(\{\be_l\}_{l=0}^{L-1}\), the vectors \(\{\ba_{l,k,m}\}\) defined by
\[
a_{l,k,m}[n] = e_l[n-m \bmod N]\exp(2\pi j k(n-m)/N)
\]
for \(l=0,\ldots, L-1\) and \(k,m,n=0,\ldots,N-1\) form a tight frame with \(\alpha = N \sum_l \|\be_l\|^2\).
\end{theorem}

\begin{proof}
Let \(\{\bs_k\}_{k=0}^{N-1}\) denote the non-unitary DFT basis while \(\bS:\mathbb{C}^N\to\mathbb{C}^N\) and \(\bD:\mathbb{C}^N\to M(\mathbb{C}^N)\) are defined by
\begin{align*}
	Sw[n] &= w[n-1 \bmod N], \\
	[D(v)w][n] &= v[n]w[n],
\end{align*}
for \(\bv,\bw \in\mathbb{C}^N\) and \(n=0,\ldots,N-1\). Using these operators \(\ba_{l,k,m} = \bS^m \bD(\be_l) \bs_k\). As \((\bS^m)^* = \bS^{-m}\) and \((\bD(\be_l))^* = \bD(\overline{\be_l})\), we have for \(\bw\in \mathbb{C}^N\)
\begin{align*}
	\sum_{k,l,m} |\langle \ba_{k,l,m}, \bw\rangle|^2 &= \sum_{k,l,m} |\langle \bS^{m} \bD(\be_l) \bs_k, \bw \rangle|^2 \\
	&= \sum_{l,m} \sum_k|\langle \bs_k, \bD(\overline{\be_l})\bS^{-m} \bw\rangle|^2.
\end{align*}
Because \(\{\bs_k\}\) is the (non-unitary) DFT basis, Plancherel's theorem implies
\begin{align*}
\sum_{k,l,m}& |\langle \ba_{k,l,m}, \bw\rangle|^2= N\sum_{l,m} \|\bD(\overline{\be_l})\bS^{-m} \bw\|^2 \\
=& N\sum_{l,m,n} |\overline{e_l[n]} w[n+m \bmod N]|^2 \\
=& N\sum_{l,n} |e_l[n]|^2 \sum_m |w[n+m \bmod N]|^2.
\end{align*}
Since \(n+m\bmod N\) ranges from \(0\) to \(N-1\),
\begin{align*}
\sum_{k,l,m} |\langle \ba_{k,l,m}, \bw\rangle|^2 &= N\sum_{l,n} |e_l[n]|^2 \|\bw\|^2 \\
&= \left(N \sum_l \|e_l\|^2 \right)\|\bw\|^2. \qedhere
\end{align*}
\end{proof}
It is an immediate corollary that if \(\|\be_l\| = (NL)^{-1/2}\) for all \(l\) then the vectors \(\{\ba_{k,l,m}\}\) form a Parseval frame. Notably, the fact that \(\{\ba_{l,k,m}\}\) is a tight frame can also be derived by viewing it as a multi-window STFT.  More importantly, the conditions on the envelopes are minimal: even a set of unrelated envelopes will admit a frame under this procedure.

With regards to representation, the defining characteristic of tight frames is that the vector \(\bw\) can be reconstructed via the formula \(\bw = \frac{1}{\alpha} \bA \bA^* \bw\) where \(\bA\) is the synthesis matrix defined by \(A[n,l,k,m] = a_{l,k,m}[n]\) \cite{framesforundergraduates}.  We refer to \(\bA^*\) as the analysis operator and note that in the case of Parseval frames \(\bA^*\) serves as the frame's right-inverse.  As \(\bA\) and \(\bA^*\) are very large, direct computation of the matrix products expensive. However both analysis and synthesis can be sped up significantly using the FFT.  If we define the vector \(\bc_{k,l}\) as \(c_{k,l}[m] = A^*w[k,l,m]\) one can show, after some matrix algebra, that
\begin{align}
\bc_{k,l} & = \bF^* \bD( \bS^k \bF \bH(\be_l))  \bF\bw, \label{eq:analysis} \\
\bw &= \frac{1}{\alpha} \bF^* \sum_{k,l} \bD(\bS^{-k} \bF \be_l) \bF \bc_{k,l}, \label{eq:synth}
\end{align}
where \(\bH:\mathbb{C}^N\to \mathbb{C}^N\) is the (conjugate linear) operator defined by \(Hw[n] = \overline{w[N-n \bmod N]}\). In practical terms, the formulations in \eqref{eq:analysis} and \eqref{eq:synth} admit parallelization of the time dimensions (\(m\) and \(n\)), and the underlying computations can be sped up via the FFT and GPU parallelization.

\bibliography{References}

\end{document}